\documentclass[english,nolineno]{socg-lipics-v2019}
\hideLIPIcs

\newcommand{\Li}{\mathcal{L}}

\title{Terrain Visibility Graphs: Persistence is Not Enough}

\author{Safwa Ameer}{Department of Computer Science, The University of Texas at San Antonio, San Antonio, USA}{safwa.ameer@gmail.com}{}{}

\author{Matt Gibson-Lopez}{Department of Computer Science, The University of Texas at San Antonio, San Antonio, USA }{matthew.gibson@utsa.edu}{}{Supported by the National Science Foundation under Grant No.~1733874.}

\author{Erik Krohn}{Department of Computer Science, The University of Wisconsin - Oshkosh, Oshkosh, USA }{krohne@uwosh.edu}{}{}

\author{Sean Soderman}{Department of Computer Science, The University of Texas at San Antonio, San Antonio, USA }{sean.soderman@my.utsa.edu}{}{Supported by the National Science Foundation under Grant No.~1733874.}

\author{Qing Wang}{Department of Computer Science, The University of Tennessee at Martin, Martin, USA }{qwang44@utm.edu}{}{}

\authorrunning{S. Ameer, M. Gibson-Lopez, E. Krohn, S. Soderman, and Q. Wang}

\Copyright{Safwa Ameer, Matt Gibson-Lopez, Erik Krohn, Sean Soderman, and Qing Wang}

\ccsdesc[100]{Computational Geometry}
\keywords{Terrains, Visibility Graph Characterization, Visibility Graph Recognition}

\EventEditors{Sergio Cabello and Danny Z. Chen}
\EventNoEds{2}
\EventLongTitle{36th International Symposium on Computational Geometry (SoCG 2020)}
\EventShortTitle{SoCG 2020}
\EventAcronym{SoCG}
\EventYear{2020}
\EventDate{June 23--26, 2020}
\EventLocation{Z\"{u}rich, Switzerland}
\EventLogo{socg-logo}
\SeriesVolume{164}

\begin{document}

\maketitle

\begin{abstract}
    In this paper, we consider the Visibility Graph Recognition and Reconstruction problems in the context of terrains.   Here, we are given a graph $G$ with labeled vertices $v_0, v_1, \ldots, v_{n-1}$ such that  the labeling corresponds with a Hamiltonian path $H$.   $G$ also may contain other edges.  We are interested in determining if there is a terrain $T$ with vertices $p_0, p_1, \ldots, p_{n-1}$ such that $G$ is the visibility graph of $T$ and the boundary of $T$ corresponds with $H$.   $G$ is said to be persistent if and only if it satisfies the so-called X-property and Bar-property.  It is known that every "pseudo-terrain" has a persistent visibility graph and that every persistent graph is the visibility graph for some pseudo-terrain.  The connection is not as clear for (geometric) terrains.  It is known that the visibility graph of any terrain $T$ is persistent, but it has been unclear whether every persistent graph $G$ has a terrain $T$ such that $G$ is the visibility graph of $T$.  There actually have been several papers that claim this to be the case (although no formal proof has ever been published), and recent works made steps towards building a terrain reconstruction algorithm for any persistent graph.  In this paper, we show that there exists a persistent graph $G$ that is not the visibility graph for any terrain $T$.  This means persistence is not enough by itself to characterize the visibility graphs of terrains, and implies that pseudo-terrains are not stretchable.
\end{abstract}


 
\section{Introduction}
The notion of geometric visibility plays a fundamental role in many applications such as robotics \cite{corke2011robotics,niku2001introduction} and shortest path computation in the presence of obstacles \cite{lozano1979algorithm}.  One of the most fundamental data structures in visibility is the \textit{visibility graph} (VG).  Let $P$ be a simple polygon in the plane with $n$ vertices labeled $p_0, \ldots, p_{n-1}$ following the boundary of $P$ in "counter-clockwise" order.  $P$ partitions the plane into two sets: "inside $P$" and "outside $P$".  We say two vertices $p_i$ and $p_j$ see each other if and only if the line segment $\overline{p_ip_j}$ does not intersect the "outside $P$" region.  The VG $G$  of $P$ has a vertex $v_i$ for each point of $p_i$, and $\{v_i,v_j\}$ is an edge in $G$ if and only if $p_i$ and $p_j$ see each other in $P$.

Given a simple polygon $P$, computing its VG in polynomial-time is a fairly trivial matter; however, if we are given a graph $G$, determining if it is the VG for some simple polygon has remained a tantalizing open problem for over 30 years.  Along these lines, there are three main VG problems that have received much attention: 1) characterization, 2) recognition, and 3) reconstruction.  In the \textit {visibility graph characterization} problem, we seek to define a set of necessary and sufficient conditions that all VGs must satisfy.  In the \textit{visibility graph recognition} problem, we seek to design an algorithm that, given a graph $G$, determines if there is a simple polygon $P$ such that $G$ is the VG of $P$.  In the \textit{visibility graph reconstruction} problem, we are given a VG $G$ and we wish to reconstruct a simple polygon $P$ such that $G$ is the VG of $P$.

\subsection{Previous work}
The history of simple polygon VG characterization dates back to 1988, when Ghosh gave three necessary conditions (NCs) that any VG must satisfy \cite{Ghosh1988}.  Shortly after, Everett and Corneil \cite{Everett1990,EverettC95} showed a counterexample to the sufficiency of NCs 1-3; that is, they gave an example of a graph that satisfies NCs 1-3 but is not the VG of any simple polygon.  Everett \cite{Everett1990} also showed that a NC might need to be strengthened to rule this example out.  Srinivsraghavan and Mukhopadhyay \cite{SrinivasaraghavanM94} showed that a strengthening of this NC was in fact necessary, but a counterexample given by Abello, Lin, and Pisupati \cite{alp1992} showed that more NCs would be needed to complete the characterization.  In 1997, Ghosh \cite{Ghosh97} gave a fourth NC that circumvents the latest counterexample, but in 2005 Streinu gave an example of a graph that satisfies the four NCs but is not a VG for any simple polygon \cite{Streinu05}. 

Unfortunately, it is not known if simple polygon VG recognition is in NP.  Even for special cases, characterization and recognition results have only been given in the extremely restricted special cases of simple polygons such as ``spiral'' polygons \cite{EverettC95} and ``tower polygons'' \cite{ChoiSC95}.  

\subsection{Pseudo-visibility}
Faced with the complexity of understanding simple polygon VGs, O'Rourke and Streinu \cite{ORourkeS97} turned their attention to \textit{pseudo-polygons}, a generalization of simple polygons where visibility is determined by a set of curves in the plane called pseudo-lines.  An arrangement of \textit{pseudo-lines} $\Li$ is a collection of simple curves, each of which separates the plane, such that each pair of pseudo-lines in $\Li$ intersects at exactly one point, where they cross.  Given a set of $n$ points in the plane and a set of pseudo-lines $\Li$ such that every pair of points has a pseudo-line that contains them, a \textit{pseudo-polygon} is determined similarly to a standard Euclidean simple polygon except that visibility is defined using $\Li$ instead of straight line segments.  Note that every simple polygon is a pseudo-polygon, where $\Li$ is a set of straight line segments.  Streinu showed that there are pseudo-polygons that cannot be \textit{stretched} into a simple polygon \cite{Streinu05}.  That is, there is a pseudo-polygon such that its VG is not the VG for any simple polygon.

In 1997, O'Rourke and Streinu \cite{ORourkeS97} gave a characterization of \textit{vertex-edge} VGs of pseudo-polygons.  In this setting, for any vertex $v$ we are told which \textit{edges} $v$ sees rather than which vertices it sees.  Unfortunately this does not extend to the desired characterization of regular VGs, as O'Rourke and Streinu showed that vertex-edge VGs encode more information about a pseudo-polygon than a regular VG \cite{ORourke1998105}.  More recently, Gibson, Krohn, and Wang gave the desired characterization of the VGs of pseudo polygons \cite{gkw15} which has very recently been extended to a polynomial-time recognition and reconstruction algorithm \cite{pseudoRecog}. 

\subsection{The visibility graphs of terrains}

One geometric structure that has gathered a lot of attention in the computational geometry community is the terrain.  A \textit{terrain} $T$ is an x-monotone (a vertical line intersects it at most once) polygonal chain in the plane.  Let $T$ be a terrain with points labeled $p_0, \ldots, p_{n-1}$ from left to right.  Let $p_i^x$ denote the x-coordinate of the point $p_i$ on $T$.  Note that due to monotonicity, we have $p_i^x < p_{i+1}^x$ for each $i \in \{0,\ldots,n-2\}$.  We say points $p_i$ and $p_j$ see each other if and only if the open line segment $\overline{p_ip_j}$ lies completely above $T$.  Given this definition of vision, one can define the VG of a terrain similarly to that of a simple polygon.  

Abello et al.~\cite{abello1993combinatorial} studied so-called ``convex fans'' which is essentially a simple polygon $P$ that can be decomposed into a terrain $T$ and one additional point $p^*$ such that $p^*$ sees every point of $T$ (the boundary of $P$ uses the boundary of $T$ as well as the line segments $\overline{p^*p_0}$ and $\overline{p^*p_{n-1}}$).  They show that every simple polygon can be decomposed into some number of convex fans, and therefore a potential strategy of tackling the simple polygon problem is to take such a decomposition and handle the fans individually.  Since $p^*$ sees every point of the convex fan, the complexity in understanding the convex fan lies almost entirely with the analysis of the ``terrain portion'' of the convex fan.  

\subsection{Persistent graphs}
With a goal towards understanding the visibility graphs of convex fans, Abello et al.~\cite{abello1995visibility} defined a notion of so-called persistent graphs and established a connection with terrain visibility graphs and persistent graphs, which we will now describe.  Suppose we are given a graph $G$ with labeled vertices $v_0, v_1, \ldots v_{n-1}$ such that $\{v_i, v_{i+1}\}$ is an edge for each $i \in \{0, 1, \ldots, n-2\}$ (i.e., the labeling gives a Hamiltonian path).  Let $H$ denote this Hamiltonian Path.  $G$ also may contain other edges.  We are interested in determining if there is a terrain $T$ with points $p_0, p_1, \ldots, p_{n-1}$ such that $G$ is the visibility graph of $T$ and the boundary of $T$ corresponds with $H$.  

$G$ is said to be \textit{persistent} if and only if it satisfies the following two properties:

\begin{itemize}
    \item \textbf{X-property:} for any set of four distinct integers $a,b,c,d \in \{0,\ldots, n-1\}$ such that $a < b < c < d$, if $\{v_a,v_c\}$ and $\{v_b,v_d\}$ are edges in $G$ then $\{v_a,v_d\}$ is also an edge in $G$.
    
    \item \textbf{Bar-property}: for every edge $\{v_i, v_k\}$ in $G$ such that $k \geq i+2$, there exists a $j \in (i,k)$ such that $\{v_i,v_j\}$ and $\{v_j,v_k\}$ are edges in $G$.
\end{itemize}

Abello et al.~\cite{abello1995visibility} showed that for any terrain $T$, its visibility graph is persistent (albeit for a slightly different definition of persistence), and Evans and Saeedi \cite{EvansS15} showed it for the definition of persistence being used here.  

\begin{figure}
\centering
\begin{tabular}{c@{\hspace{0.1\linewidth}}c}

\includegraphics[scale=.2]{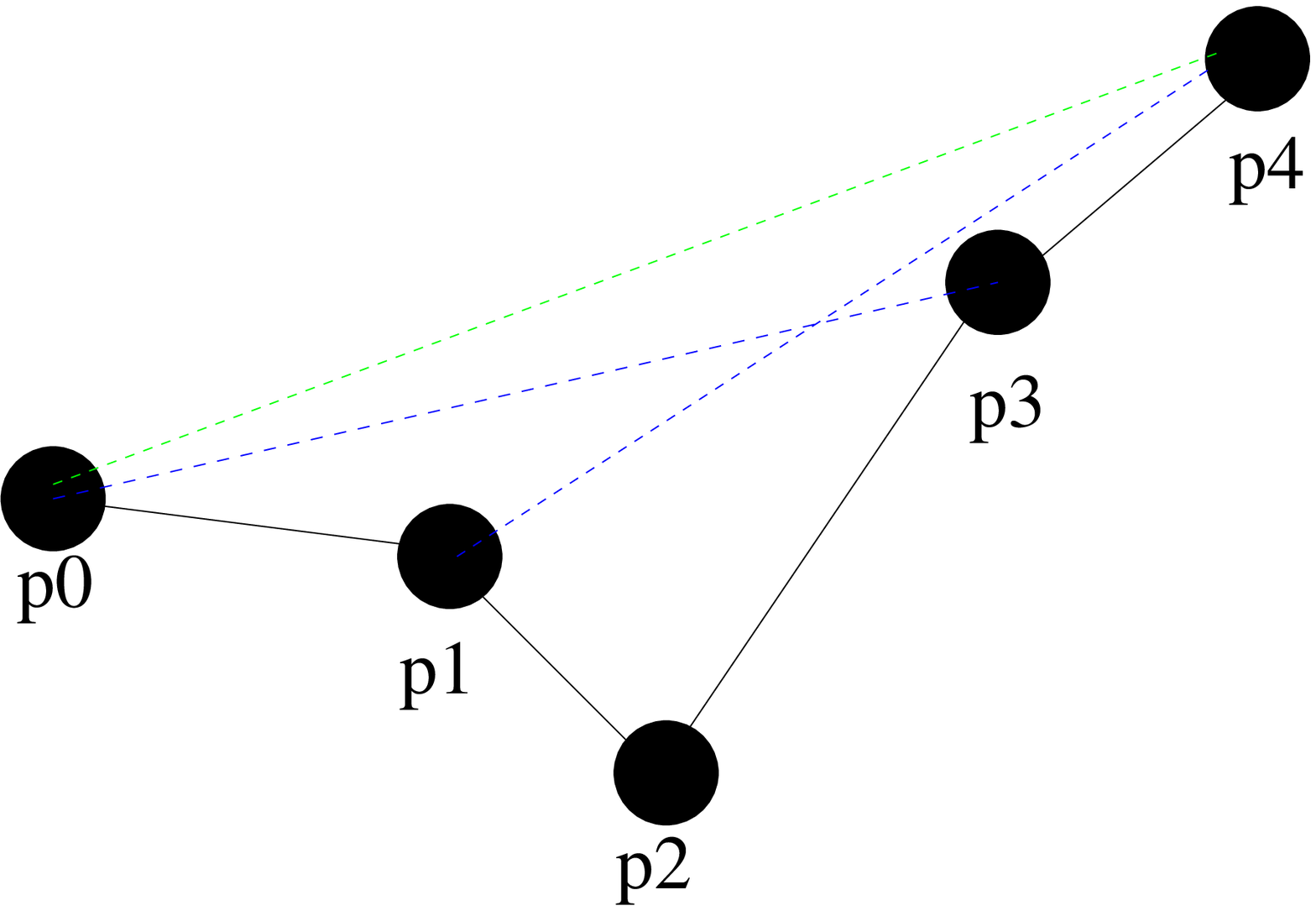}&
\includegraphics[scale=.2]{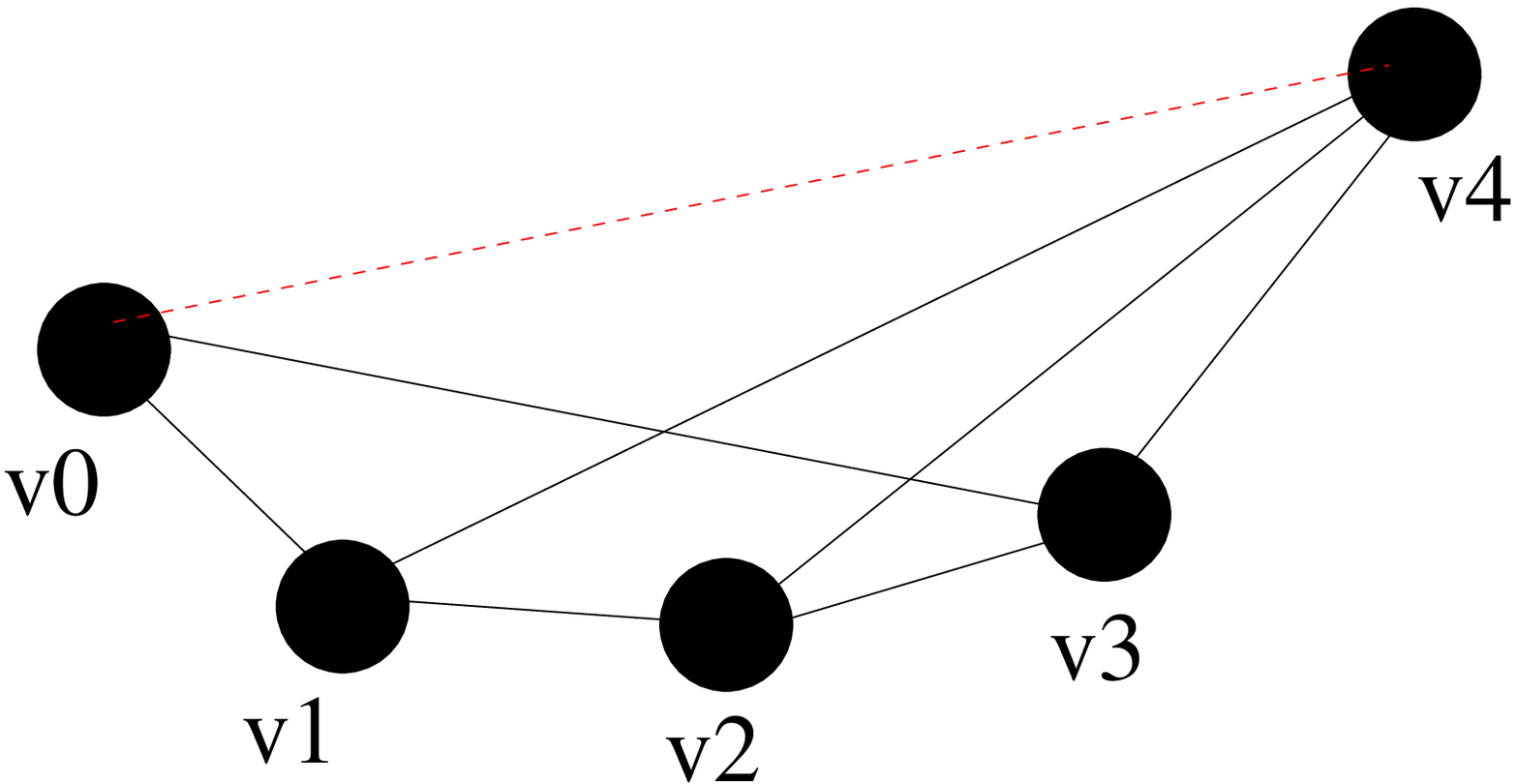}
\\
(a) & (b) 
\end{tabular}
\caption{\footnotesize An illustration of the X-property.}
\label{fig:xProp}
\end{figure}

We now will help develop intuition for these properties (see \cite{EvansS15} for a formal proof).  For the X-property (sometimes referred to as the "order claim"), consider Figure \ref{fig:xProp}.  In part (a), we have a terrain such that: (1) $p_0$ sees $p_3$, and (2) $p_1$ sees $p_4$ (the blue dotted lines).  Therefore no vertex between $p_0$ and $p_4$ is strictly above either of the blue dotted lines.  Then the line segment connecting $p_0$ and $p_4$ will be "above" the blue dotted lines and therefore $p_0$ must see $p_4$.  So now consider the graph in part (b).  If the edges $\{v_0,v_3\}$ and $\{v_1,v_4\}$ are in the graph but $\{v_0,v_4\}$ is not an edge in the graph then it cannot be the visibility graph of a terrain.  

\begin{figure}
\centering
\begin{tabular}{c@{\hspace{0.1\linewidth}}c}

\includegraphics[scale=.2]{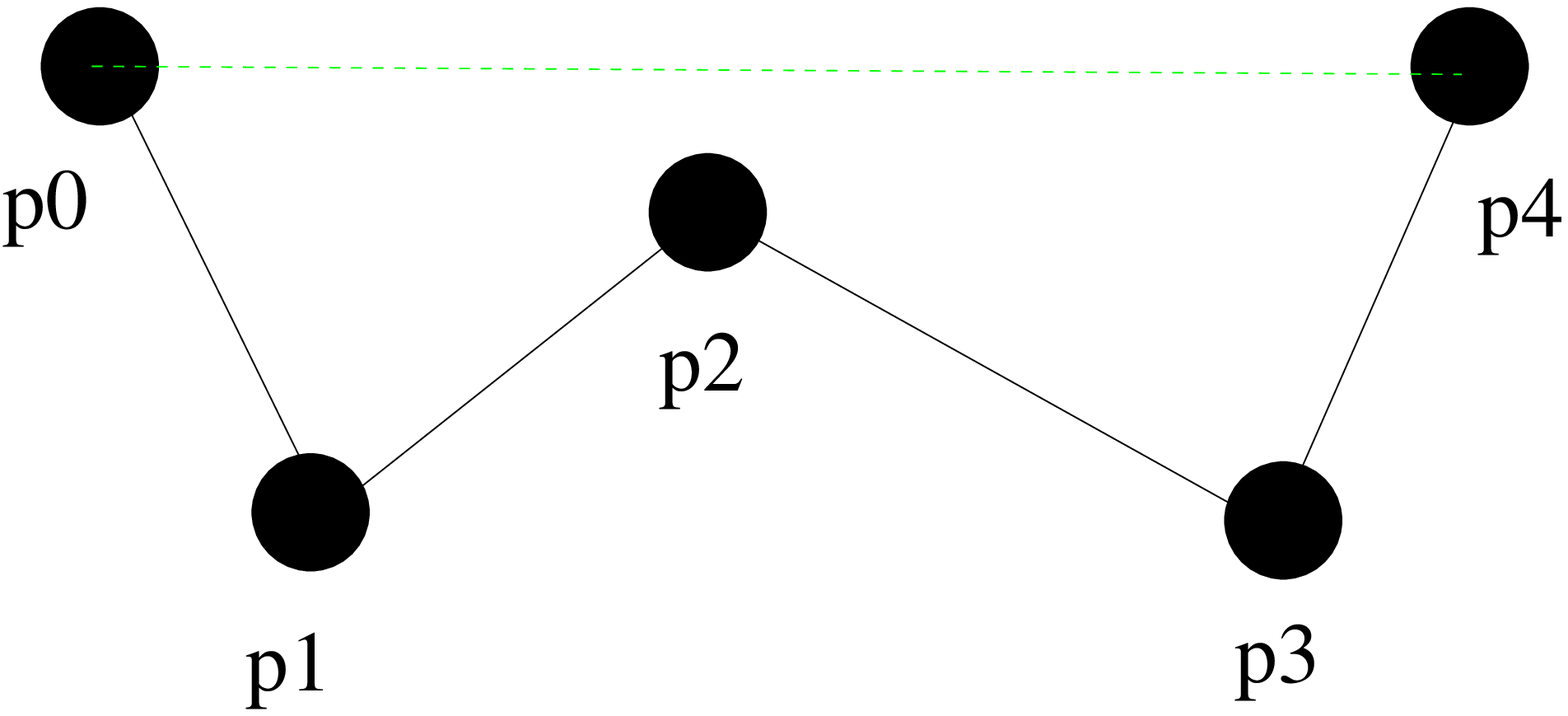}&
\includegraphics[scale=.2]{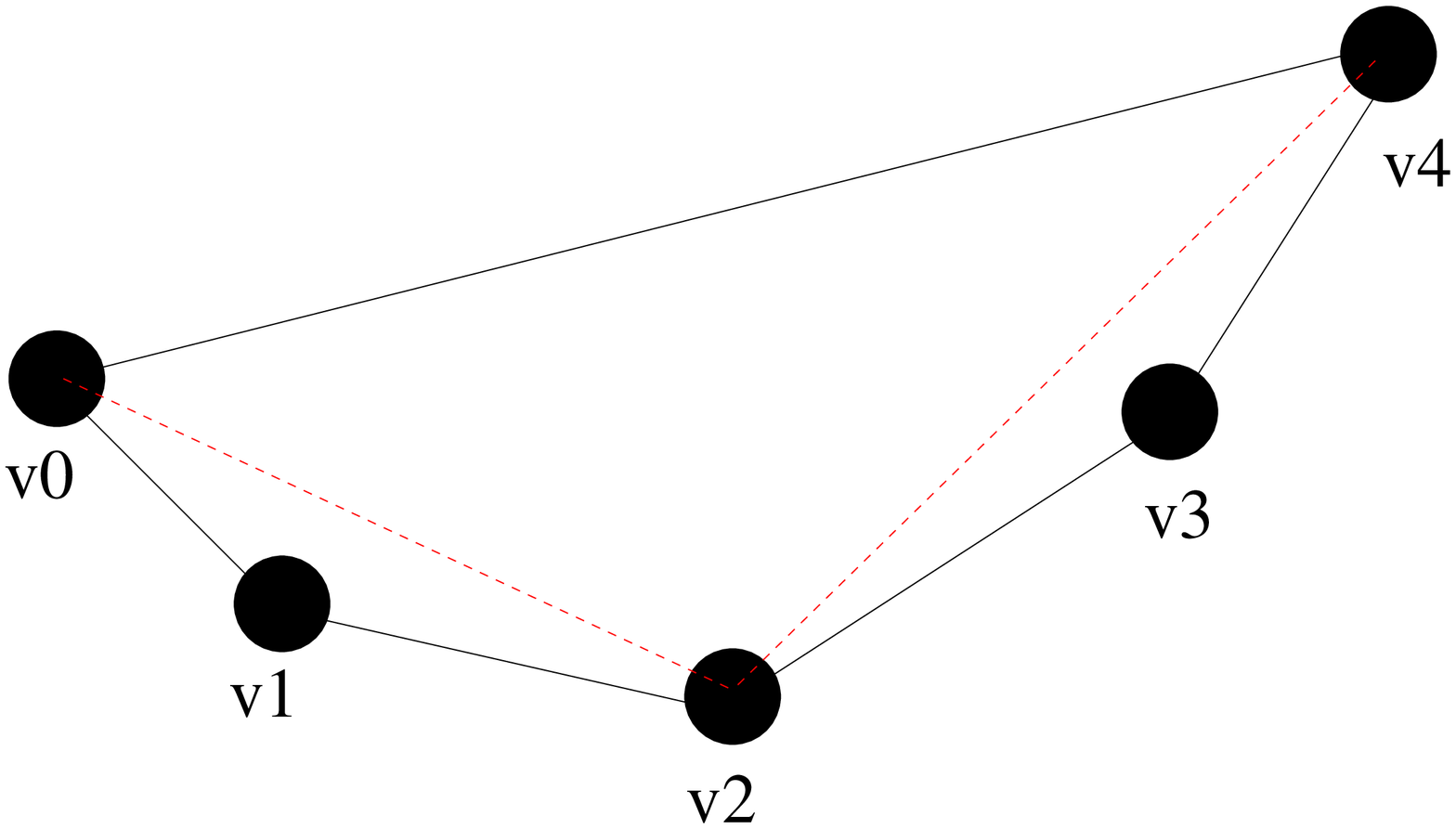}

\\
(a) & (b)
\end{tabular}
\caption{\footnotesize An illustration of the Bar-Property.}
\label{fig:barProp}
\end{figure}

For the Bar-property, see Figure \ref{fig:barProp}.  In the terrain in part (a), we have that $p_0$ sees $p_4$.  $p_1$ sees $p_0$, but it doesn't see $p_4$ because it is blocked by $p_2$.  Then it must be that $p_2$ also sees $p_0$.  Since $p_2$ also sees $p_4$ so we are done.  In general, if $p_i$ sees $p_k$, then $p_{i+1}$ must see $p_i$.  If $p_{i+1}$ also sees $p_k$ we are done, so suppose it doesn't see $p_k$ because there is some point $p_j$ for $j\in\{i+2, \ldots, k-1\}$ such that $p_{i+1}$ sees $p_j$, and $p_j$ is over the line segment $\overline{p_{i+1},p_k}$.  $p_j$ must see $p_i$, and if it sees $p_k$ we are done.  Otherwise we repeat this argument with the point that blocks $p_j$ from $p_k$, and eventually we find a point that sees both $p_i$ and $p_k$.  Therefore if the graph in part (b) only contains the black edges, it cannot be the visibility graph of a terrain, as the graph implies that $p_0$ should see $p_4$ but no other point between them should see both $p_0$ and $p_4$.

Abello et al.~\cite{abello1995visibility} showed a one-to-one correspondence between the VGs of pseudo-terrains (terrains using pseudo-lines to define visibilities rather than straight line segments) and persistent graphs.  That is, they show that the VG of any pseudo-terrain is persistent, and they show that any persistent graph has a pseudo-terrain and give a polynomial-time algorithm to reconstruct it.  Evans and Saeedi \cite{EvansS15} give a simpler proof (and a faster reconstruction algorithm) of the same result.

It has remained an open problem to show that persistent graphs and the visibility graphs of (geometric) terrains are exactly the same set (i.e., to show that $G$ is a persistent graph if and only if there is a terrain $T$ such that $G$ is the visibility graph of $T$). Several papers have made progress towards giving a reconstruction algorithm that can take a persistent graph $G$ as input and construct a terrain $T$ such that $G$ is the visibility graph of $T$.  In fact, there are papers \cite{abello1993combinatorial,abello2004majority} that claim that there exists such a reconstruction algorithm although a formal proof of this has not been published.  Evans and Saeedi \cite{EvansS15} state that they ideally would like to reconstruct a terrain from a persistent graph but that it seems difficult.  Most of the previous attempts to reconstruct terrains from a persistent graph involves an iterative placement of the points of the terrain (e.g., determining the x and y coordinates of the points of the terrain from left to right).

\subsection{Our contribution}
The main result of this paper is to prove that these two classes of graphs are in fact \textit{not} the same.  

\begin{theorem}
There is a persistent graph $G$ such that there is no terrain $T$ such that $G$ is the visibility graph of $T$.
\label{thm:mainResult}
\end{theorem}

We obtain this result by introducing a new style of reconstruction algorithm.  We show that if one can compute a set of feasible x-coordinates for the points of the terrain, then the y-coordinates can be computed via linear programming (LP).  Using standard LP analysis techniques, we identify a seven-vertex, persistent graph $G'$ that must have its x-coordinates chosen carefully in order to be able to reconstruct a terrain with $G'$ as its visibility graph.   We then build a graph $G^*$ that has thirty-five vertices which can be partitioned into five ``copies" of $G'$.  In order to represent $G^*$ as a terrain, we would need to pick the thirty-five x-coordinates in a way where each ``copy" of $G'$ has its condition satisfied, and we show that this is not possible.

Since $G^*$ is persistent, it is the visibility graph of some pseudo-terrain, and therefore our result also is a proof that pseudo-terrains are not stretchable.

\subsubsection{Organization of the paper}  In Section \ref{sec:reconstruct}, we describe our LP-based reconstruction algorithm.  In Section \ref{sec:pickyGraph}, we give our graph $G'$ and show that it requires very specifically chosen x-coordinates in order to be realizable as a terrain.  This critically uses our new LP-based reconstruction approach.  In Section \ref{sec:noTerrain}, we give our persistent graph $G^*$ and prove that there is no terrain that has it as its visibility graph.  We give a conclusion and some open problems in Section \ref{sec:conc}.

\section{Reconstructing terrains via linear programming}
\label{sec:reconstruct}
Let $G$ be a persistent graph with vertices $v_0, \ldots v_{n-1}$.  For any terrain $T$ with points $p_0, \ldots, p_{n-1}$, we let $p_i^x$ denote the x-coordinate of $p_i$.  Let $X = (x_0, x_1, \ldots x_{n-1})$ be a vector of real numbers such that $x_i < x_{i+1}$ for each $i \in \{0, 1, \ldots n-2\}$, and let $\mathcal{T}(G,X)$ be the set of all terrains $T$ with $n$ points such that:
\begin{enumerate}\item$p_0^x = x_0, p_1^x = x_1, \dots p_{n-1}^x = x_{n-1}$ (i.e., it is the set of all terrains that have x-coordinates that correspond with $X$).

\item The boundary of $T$ corresponds to the Hamiltonian path $H$.

\item $G$ is the visibility graph of $T$.
\end{enumerate}
For any two integers $i,j \in \{0, \ldots, n-1\}$ such that $i<j$, let $d_{i,j} := |x_i - x_j|$.  Intuitively, for a terrain $T \in \mathcal{T}(G,X)$, $d_{i,j}$ is the distance between the x-coordinates of $p_i$ and $p_j$.

We will now show that given $G$ and $X$, we can determine in polynomial-time if there is a terrain in $\mathcal{T}(G,X)$, and moreover if $\mathcal{T}(G,X) \neq \emptyset$ then we can compute in polynomial-time a feasible set of y-coordinates for some terrain $T \in \mathcal{T}(G,X)$.  This algorithm is via a reduction to linear programming (LP) where the variables of the LP are the y-coordinates of the points of the terrain $T$.  We show that given a fixed set of x-coordinates, we can model all of the visibility constraints that $T$ must satisfy as inequalities that are linear in the y-coordinates of the points of $T$.  It is not immediately obvious blocking constraints can be modeled as linear constraints (i.e, if $\{v_i,v_j\}$ is not an edge in $G$, ensuring that the y-coordinates are computed so that the points $p_i$ and $p_j$ do not see each other in $T$), but we will show that we can in fact model this as a linear constraint.

\begin{figure}[ht]
\centering
\begin{tabular}{c@{\hspace{0.1\linewidth}}c@{\hspace{0.1\linewidth}}
c}
\includegraphics[scale=0.75]{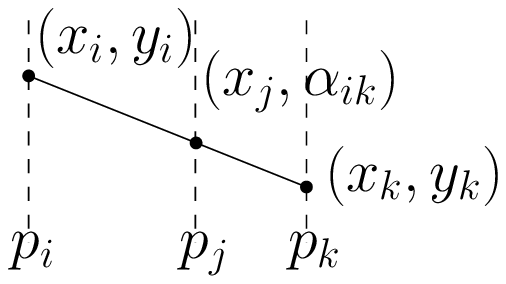}&
\includegraphics[scale=0.15]{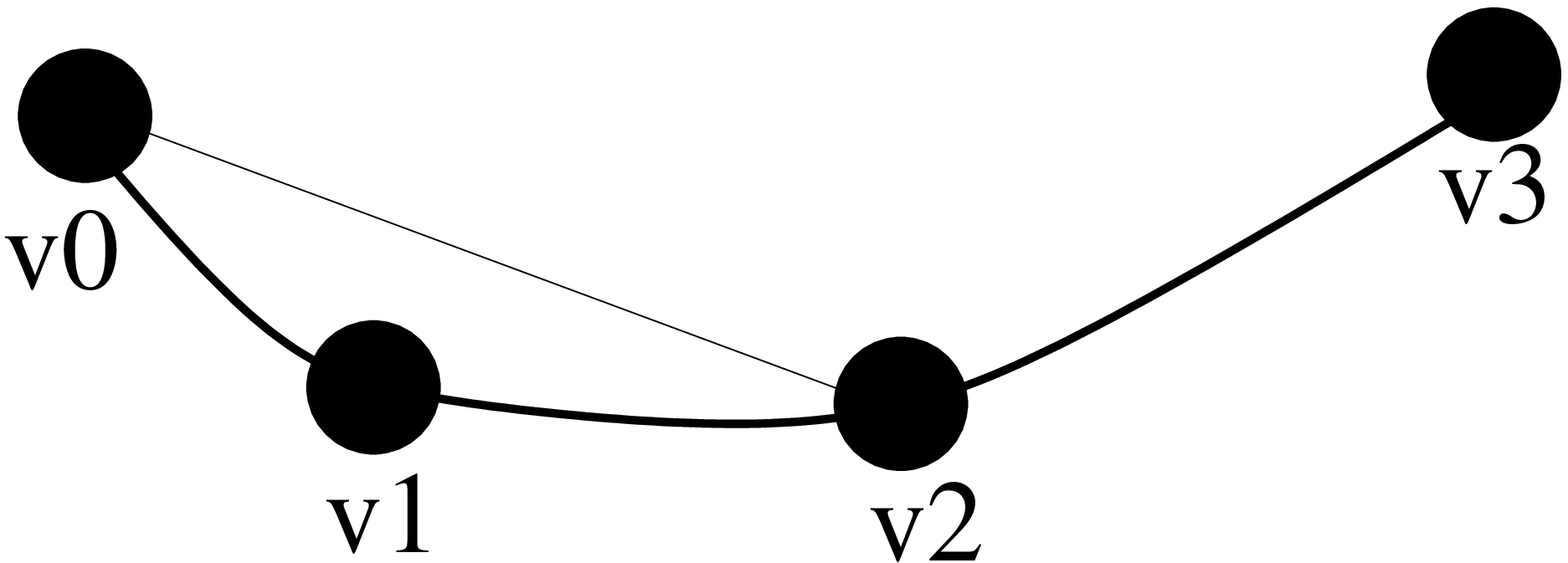}&
\includegraphics[scale=0.15]{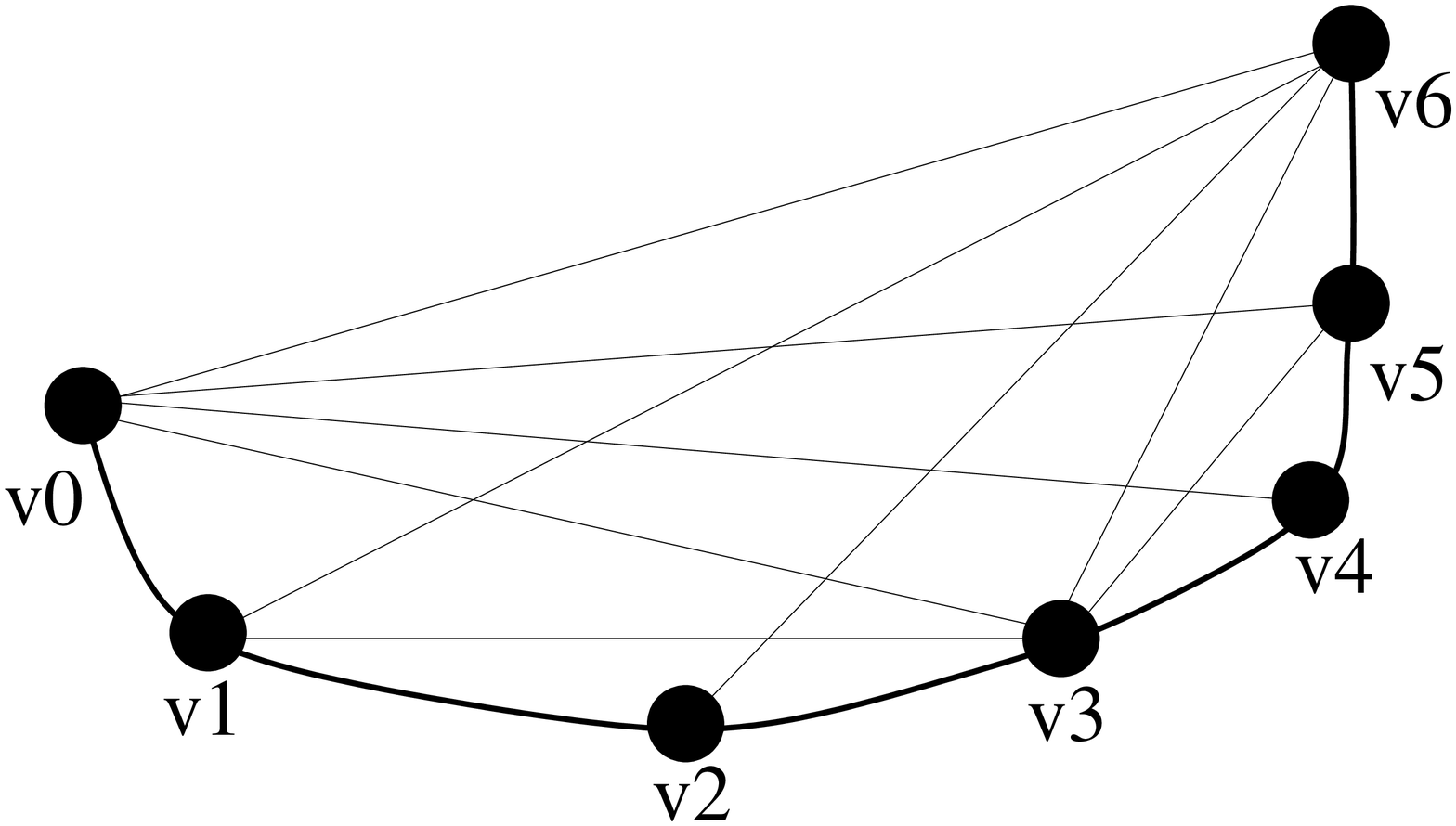}
\\
 (a) & (b) & (c)
\end{tabular}
 \caption{\footnotesize (a) LP constraint illustration.  (b) A sample terrain VG. (c) The VG G'.
 }
 \label{fig:construction}
\end{figure}

First let us consider a visibility constraint: let $\{v_i,v_k\}$ be an edge in $G$.  We must ensure that the y-coordinates $y_i$ and $y_k$ for $p_i$ and $p_k$ respectively are such that the line segment $\overline{p_ip_k}$ ``stays above" $T$.  We can ensure this, by enforcing that for every $j \in (i,k)$, we choose the y-coordinate $y_j$ such that $p_j$ is underneath $\overline{p_ip_k}$.  Let $\alpha^j_{ik}$ denote the $y$-coordinate of the intersection of $\overline{p_ip_k}$ and the vertical line $x = x_j$ (as illustrated in Figure \ref{fig:construction} (a)).  It is easy to see that $\alpha^j_{ik} = \frac{d_{j,k}\cdot y_i + d_{i,j}\cdot y_k}{d_{i,k}}$, a linear function of $y_i$ and $y_k$ since $d_{i,j}, d_{j,k}$, and $d_{i,k}$ are functions of the constant $x$-coordinates. Therefore the \textit{visibility constraint} $y_j < \alpha^j_{ik}$ is a linear inequality.  In our LP, we will write the constraint as $d_{j,k}\cdot y_i - d_{i,k}\cdot y_j + d_{i,j}\cdot y_k \geq \epsilon$ where $\epsilon$ is a positive constant.  Note that $\{v_i,v_k\}$ can have many constraints in the LP associated with it (although some of them may be redundant and can be removed without affecting the set of feasible solutions to the LP, more on this later).

Now suppose $v_i$ and $v_k$ are such that $\{v_i,v_k\}$ is not an edge in $G$.  Then we must enforce that the corresponding points $p_i$ and $p_k$ do not see each other in $T$.  This means that $\overline{p_ip_k}$ must cross under the terrain $T$.  We can do this by enforcing that \textit{some} point $p_j$ between $p_i$ and $p_k$ has its y-coordinate chosen to be large enough so that it is above $\overline{p_ip_k}$.  Unfortunately the notion that \textit{some} point must be over $\overline{p_ip_k}$  cannot directly be represented as a linear constraint (whereas in the previous case it had to be that \textit{every} point must be under $\overline{p_ip_k}$).  However we can see that by employing an analysis similar to the so-called \textit{designated blocker} from the analysis of pseudo-polygon visibility graphs \cite{gkw15}, we can identify a specific point (or two) that \textit{must} be above $\overline{p_ip_k}$ which allows us to express the constraint as a linear inequality.  To find the first such point, start at $v_k$ and ``walk to the left" along $H$ towards $v_i$ and let $v_j$ be the first vertex encountered such that $\{v_i,v_j\}$ is an edge in $G$ (note that such a vertex must exist; $\{v_i,v_{i+1}\}$ is an edge in $G$).  We claim that for every $T \in \mathcal{T}(G,X)$, it must be that $p_j$ is over $\overline{p_ip_k}$.  Suppose for the sake of contradiction that $p_j$ is under $\overline{p_ip_k}$.  If there is a point $p_z$ over $\overline{p_ip_k}$ such that $z < j$, then $p_i$ doesn't see $p_j$, a contradiction, so suppose there is no such point over $\overline{p_ip_k}$. So now let $p_z$ be the first point to the right of $p_j$ that is over $\overline{p_ip_k}$.  Then it must be that $p_i$ sees $p_z$, but $\{v_i,v_z\}$ is not an edge in $G$ by definition of $p_j$, a contradiction.  So it is true that for every $T \in \mathcal{T}(G,X)$, it must be that $p_j$ is over $\overline{p_ip_k}$, and we call $p_j$ the designated blocker to block $p_i$ from $p_k$.  Therefore we can add the blocking constraint $y_j > \alpha^j_{ik}$ to our LP.  We write this constraint $-d_{j,k}\cdot y_i + d_{i,k}\cdot y_j - d_{i,j}\cdot y_k \geq \epsilon$ where $\epsilon$ is a positive constant.  We symmetrically compute the designated blocker to block $p_k$ from seeing $p_i$.  Note that this point $p_{j'}$ may not be the same point as the first designated blocker $p_j$ (but it must be that $j \leq j'$ or else $G$ violates the X-property and therefore is not persistent).  If $j' \neq j$, then we add another blocking constraint for $p_{j'}$.  We again remark that sometimes these blocking constraints are redundant and can be removed without altering the set of feasible solutions to the LP.

The choice of $\epsilon$ does not effect whether or not there is a feasible solution to the LP (as long as $\epsilon$ is positive).  If there is a solution vector $\textbf{y}$ that is feasible with right hand side $\epsilon'$, then one can obtain a feasible solution with right hand side $\epsilon$ by scaling $\textbf{y}$ by a factor of $\frac{\epsilon}{\epsilon'}$.

To illustrate our approach, consider the example VG in Figure \ref{fig:construction} (b).  We will show how we construct the LP in order to reconstruct a terrain that has this graph as its VG.  Suppose $X = (0,1,2,3)$.  First note that since $\{v_0,v_2\}$ is an edge, we need the visibility constraint $y_1 \leq \alpha^1_{0,2} = \frac{y_0 + y_2}{2}$, which we can write as $y_0 - 2y_1 + y_2 \geq 1$.  Secondly note that $p_0$ and $p_3$ do not see each other and $p_2$ is the designated blocker.  Therefore we add the blocking constraint $y_2 > \alpha^2_{0,3} = \frac{y_0 + 2y_3}{3}$, which we state as $-y_0 + 3y_2 - 2y_3 \geq 1$.  Note $p_1$ does not see $p_3$ and has designated blocker $p_2$, but this constraint is redundant with the other two constraints.  Therefore our final LP is the following: $y_0 - 2y_1 + y_2 \geq 1; -y_0 + 3y_2 - 2y_3 \geq 1$. Any feasible solution to this LP will give $y$-coordinates for a terrain $T$ such that $G$ is the VG of $T$.

One of the advantages of the LP-based approach is we can use standard LP techniques to help us determine what (if any) constraints on x-coordinates need to be satisfied in order to reconstruct the terrain (or determine that no x-coordinates are possible).  In particular, we will be using the well-known \textit{Farkas' Lemma}.  Let $m$ denote the number of constraints in our LP, and let $n$ be the number of points in the terrain.  The LP can be represented as $\textbf{Ay} \geq \textbf{b}$, where $\textbf{A}$ is an $m \times n$ matrix of coefficients, $\textbf{y} \in \mathbb{R}^n$ is the vector of y-coordinate variables of the LP, and $\textbf{b} = \{\epsilon\}^m$ for some $\epsilon > 0$.  Then Farkas' Lemma \cite{Lasserre200467} says that \textit{exactly one} of the following two statements is true: \begin{enumerate}
    \item there exists a $\textbf{y}$ satisfying $\textbf{Ay} \geq \textbf{b}$ (i.e., there exists a terrain in $\mathcal{T}(G,X)$) \item there is a $\textbf{z} \in \mathbb{R}^m$ such that $\textbf{z} \leq 0$, $\textbf{A}^\text{T}\textbf{z} \geq 0$ and $\textbf{b}^\text{T}\textbf{z} < 0$.
\end{enumerate}  Our result heavily relies on the use of Case 2 of Farkas' Lemma to determine exactly which $X$ vectors create a non-empty $\mathcal{T}(G,X)$ for a given persistent graph $G$.

\section{A picky persistent graph}
\label{sec:pickyGraph}

In this section, we will prove one of the key lemmas that leads to our result: there is a persistent graph that requires its x-coordinates to satisfy a strict inequality in order for there to be a feasible solution to the LP. The same visibility graph was analyzed in \cite{abello1993visibility} where they showed that this graph cannot be represented with ``uniform step lengths" 
(which in our context means that for any $c>0$ we have $d_{i,i+1} = c$).  While this graph has been observed in previous works, what is new in this paper is the exact requirements that the x-coordinates must satisfy in order for there to be a terrain.

Let $G'$ be the visibility graph in Figure \ref{fig:construction} (c).  A terrain that has $G'$ as its visibility graph is shown in Figure \ref{fig:gPrime}.  Consider the LP using the following constraints: (1) $p_1$ should be above $\overline{p_0p_2}$, (2) $p_3$ should be under $\overline{p_0p_4}$, (3) $p_3$ should be over $\overline{p_1p_5}$, (4) $p_3$ should be under $\overline{p_2p_6}$, (5) $p_5$ should be under $\overline{p_3p_6}$, and (6) $p_5$ should be over $\overline{p_4p_6}$.  Note that there are other constraints we aren't explicitly stating here such as $p_3$ being under $\overline{p_0p_5}$ (we will show they are redundant and adding them does not affect the feasible region of the LP; removing the redundant constraints will simplify the later analyses).  Here the number of constraints $m = 6$ and the number of points $n = 7$.  We express this LP in the form $\textbf{Ay} \geq \textbf{b}$ where \textbf{A}, \textbf{y}, and \textbf{b} are as follows:

\begin{center}

\begin{tabular}{c c c}
$\textbf{A} = 
\begin{vmatrix}
-d_{1,2} & d_{0,2} & -d_{0,1}  & 0   & 0 & 0 & 0 \\
d_{3,4} & 0          & 0    & -d_{0,4}      & d_{0,3}     & 0 & 0\\
0    & -d_{3,5}    & 0         & d_{1,5}    & 0    & -d_{1,3} & 0   \\
0     & 0 & d_{3,6} & -d_{2,6} & 0        & 0     & d_{2,3}\\
0     & 0    & 0         & d_{5,6}   & 0   & -d_{3,6} & d_{3,5}   \\
0     & 0          & 0    & 0          & -d_{5,6} & d_{4,6} & -d_{4,5}
\end{vmatrix}  $ 

&
$\textbf{y} = 
\begin{vmatrix}
y_0 \\
y_1 \\
y_2 \\
y_3 \\
y_4 \\
y_5 \\
y_6
\end{vmatrix}   $ 
&
$\textbf{b} = 
\begin{vmatrix}
\epsilon \\
\epsilon \\
\epsilon \\
\epsilon \\
\epsilon \\
\epsilon 
\end{vmatrix}  $
\end{tabular}
\end{center}

Again, $\epsilon$ is a positive constant.  Let $T(X,\textbf{y})$ denote the $n$-point terrain whose x-coordinates correspond with $X$ and y-coordinates correspond with $\textbf{y}$.  Clearly if $T$ is a terrain in $\mathcal{T}(G',X)$ then the vector of y-coordinates of its points is a feasible solution to this LP.  We will now argue that if $\textbf{y}$ is a feasible solution to this LP then $T(X,\textbf{y}) \in \mathcal{T}(G',X)$.  

\begin{lemma}
Let $\textbf{y}$ be a feasible solution to the LP.  Then the visibility graph of $T(X,\textbf{y})$ is $G'$.
\label{lem:lpFeas}
\end{lemma}

\begin{proof}
The combination of constraints 5 ($p_5$ should be under $\overline{p_3p_6}$) and 6 ($p_5$ should be over $\overline{p_4p_6}$) directly implies that the visibilities of $p_i$ and $p_j$ correctly match those given by $G'$  for $v_i$ and $v_j$ for each pair when $i,j \geq 3$.  In particular, $p_4$ must be under $\overline{p_3p_5}$ and $\overline{p_3p_6}$.

Now consider point $p_0$.  Constraint 2 ($p_3$ should be under $\overline{p_0p_4}$) implies that $p_0$ will see $p_3, p_4, p_5$, and $p_6$ as long as $p_1$ and $p_2$ do not block them.  Constraint 3 ($p_3$ should be over $\overline{p_1p_5}$) ensures that $p_1$ will be under $\overline{p_0p_3}$ and then Constraint 1 ($p_1$ should be above $\overline{p_0p_2}$) implies $p_2$ is under $\overline{p_0p_3}$ and $\overline{p_1p_3}$.  Therefore $p_0$ will correctly see $p_3, p_4, p_5$, and $p_6$.  Moreover, Constraint 1 directly implies that $p_0$ will not see $p_2$, and therefore all visibilities corresponding to $p_0$ are correct.  

The fact that $p_2$ is under $\overline{p_1p_3}$ implies that $p_1$ and $p_3$ correctly see each other.  Constraint 4 ($p_3$ should be under $\overline{p_2p_6}$) implies that $p_2$ will correctly see $p_6$ given the earlier configurations of $p_4$ and $p_5$.  So using the fact that the visibility graph of any terrain satisfies the X-property, we can see that $p_1$ correctly sees $p_6$ (applying the X-property with $a = 1, b = 2, c = 3$, and $d = 6$).

Finally we need that $p_i$ does not see $p_j$ for $i\in \{1,2\}$ and $j \in \{4,5\}$.  $p_1$ does not see $p_5$ as directly implied by Constraint 3, and we already showed the following: $p_2$ is under $\overline{p_1p_3}$, $p_4$ is under $\overline{p_3,p_5}$.  This implies the remaining three pairs of points correctly do not see each other.
\end{proof}

The following lemma uses Farkas' Lemma to determine requirements on $X$ (which in turn determines $\textbf{A}$) in order to have $\mathcal{T}(G',X) \neq \emptyset$.

\begin{lemma}
There is a terrain $T \in \mathcal{T}(G',X)$ if and only if $X$ satisfies $d_{0,1}d_{2,3}d_{3,4}d_{5,6} > d_{1,2}d_{4,5}d_{0,3}d_{3,6}$.
\label{lem:mainIneq}
\end{lemma}

\begin{proof}



Suppose that $X$ satisfies $d_{0,1}d_{2,3}d_{3,4}d_{5,6} > d_{1,2}d_{4,5}d_{0,3}d_{3,6}$.  Let $\epsilon$ (which appears in \textbf{b}) be the minimum of $d_{3,5}(d_{0,1}d_{2,3}d_{3,4}d_{5,6} - d_{1,2}d_{4,5}d_{0,3}d_{3,6})$ and $d_{3,5}(d_{0,1}d_{3,4}d_{5,6}d_{3,5} + d_{3,4}d_{5,6}d_{0,2}d_{3,6}
    + d_{1,2}d_{5,6}d_{3,5}d_{3,6} + d_{3,4}d_{5,6}d_{3,5}d_{3,6}
    + d_{1,2}d_{0,3}d_{3,6}d_{3,5})$.  Note that $\epsilon$ is strictly positive given our assumption on $X$.  We show that the following vector $\textbf{y}$ is a feasible solution to the LP (work shown in the appendix):

\[\textbf{y} =
    \begin{vmatrix}
    d_{0,1}d_{2,3}d_{3,5}d_{5,6} + d_{0,3}d_{3,5}d_{0,1}d_{2,3} + 
    d_{0,3}d_{3,5}d_{0,1}d_{4,5} +
    d_{4,5}d_{0,3}d_{0,2}d_{3,6} +
    d_{0,3}d_{3,5}d_{4,5}d_{3,6}\\
    \\
    
    d_{1,2}d_{4,5}d_{0,3}d_{3,6} - d_{0,1}d_{2,3}d_{3,4}d_{5,6} \\
    \\
    
    -d_{2,3}d_{5,6}(d_{3,4}d_{0,2} + d_{3,5}(d_{1,2} + d_{3,4}))
    - d_{1,2}d_{0,3}d_{3,5}(d_{2,3} + d_{4,5})\\
    \\
    
    0\\
    \\
    
    -d_{0,1}d_{3,4}d_{3,5}(d_{2,3} + d_{4,5}) - d_{4,5}d_{3,6}
    (d_{3,4}(d_{0,2}+d_{3,5}) + d_{1,2}d_{3,5})\\
    \\
    0\\
    \\
    
    d_{0,1}d_{3,4}d_{5,6}d_{3,5} + d_{3,4}d_{5,6}d_{0,2}d_{3,6} + d_{1,2}d_{5,6}d_{3,5}d_{3,6} + d_{3,4}d_{5,6}d_{3,5}d_{3,6} + d_{1,2}d_{0,3}d_{3,6}d_{3,5}
    \end{vmatrix}
\]

Now suppose $X$ is such that $d_{0,1}d_{2,3}d_{3,4}d_{5,6} \leq d_{1,2}d_{4,5}d_{0,3}d_{3,6}$.  We will show that $\mathcal{T}(G',X) = \emptyset$ by using Farkas' Lemma.  In particular, we show that there is a vector $\textbf{z} \in \mathbb{R}^m$ such that $\textbf{z} \leq 0$, $\textbf{A}^\text{T}\textbf{z} \geq 0$, and $\textbf{b}^\text{T}\textbf{z} < 0$ for every $\epsilon > 0$. Our vector $\textbf{z}$ is as follows:

\begin{center}
\[ \textbf{z} = \begin{vmatrix}
    \dfrac{-d_{3,4}d_{5,6}}{d_{1,2}d_{0,3}} \\
    \\
    
    \dfrac{-d_{5,6}}{d_{0,3}} \\
    \\
    
    \dfrac{-d_{3,4}d_{5,6}d_{0,2}}{d_{2,1}d_{0,3}d_{3,5}} \\
    \\
    
    \dfrac{-d_{0,1}d_{3,4}d_{5,6}}{d_{1,2}d_{0,3}d_{3,6}}\\
    \\
    \dfrac{d_{0,1}d_{2,3}d_{3,4}d_{5,6} - d_{1,2}d_{4,5}d_{0,3}d_{3,6}}{d_{1,2}d_{0,3}d_{3,6}d_{3,5}}\\
    \\
    -1
    \end{vmatrix}
    \]
    \end{center}

Note that the next-to-last entry is at most 0 due to the assumption on $X$, and the rest are strictly negative for all $X$.  Therefore it immediately follows that $\textbf{z} \leq 0$ and $\textbf{b}^\text{T}\textbf{z} < 0$ for every $\epsilon > 0$.  We complete the proof by showing that $\textbf{A}^\text{T}\textbf{z}$ is a zero vector (work shown in the appendix).
\end{proof}

\begin{figure}[ht]
\centering

\includegraphics[scale=0.43]{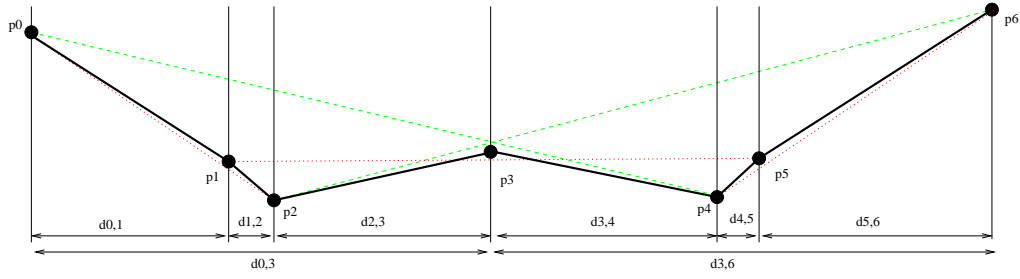}

 \caption{\footnotesize A terrain whose VG is $G'$.
 }
 \label{fig:gPrime}
\end{figure}

We remark that Lemma \ref{lem:mainIneq} can illustrate the difficulty in designing an algorithm that reconstructs the terrain from left to right, placing the points of the terrain one at a time.  Let $G''$ be the subgraph of $G'$ induced by the first six vertices $\{v_0, \ldots, v_5\}$.  It is not hard to see that $G''$ can be reconstructed using any vector of six, increasing x-coordinates.  Suppose we take such a reconstruction and then try to extend the reconstruction to handle all of $G'$.  If we reconstructed $G''$ using, say, $x_i = i$ for each $i \in \{0, \ldots, 5\}$ (implying that $d_{i,i+1} = 1$ for each $i \in \{0, \ldots, 4\}$), one can see that \textit{every} choice of $x_6$ such that $x_6 > x_5$ will violate the inequality stated in Lemma \ref{lem:mainIneq} (note that the choice of $x_6$ impacts the $d_{5,6}$ term on the left side and impacts the $d_{3,6}$ term on the right side).  This implies that a left-to-right style approach may need to shift both the x-coordinates and y-coordinates of the previously-placed points to accommodate the new point.

\section{A persistent graph that is not a terrain visibility graph}
\label{sec:noTerrain}

We are now ready to prove our main result of the paper, that there is a persistent graph $G^*$ such that there is no terrain $T$ such that $G^*$ is the visibility graph of $T$.  The adjacency matrix of $G^*$ is given in Figure \ref{fig:adjMatrix}.  There are 35 vertices in $G^*$, listed from left to right along the ``horizontal axis" of the graph.  The naming convention that we are using in this graph partitions the vertices into five color groups, each color containing seven vertices.  There is green $(g_0,\ldots,g_6)$, red $(r_0, \ldots, r_6)$, blue $(b_0, \ldots, b_6)$, magenta ($m_0, \ldots, m_6)$, and yellow $(y_0, \ldots y_6)$.  The key observation about each of these color classes is that the subgraph of $G^*$ induced by each of the color classes is exactly the graph $G'$ used in Lemma \ref{lem:mainIneq}, and moreover the designated blockers are exactly the same.  For example, $g_1$ must be over $\overline{g_0g_2}$, because $g_0$ doesn't see any point between $g_1$ and $g_2$ (including points of different colors) and $g_2$ doesn't see any point between $g_0$ and $g_1$. This implies that in order to obtain a terrain $T$ that has $G^*$ as its visibility graph, the x-coordinates must be chosen so that each of the 5 color classes satisfy the inequality of Lemma \ref{lem:mainIneq}, and we will show that this is not possible.

\begin{figure}

\centering

\includegraphics[scale=0.37]{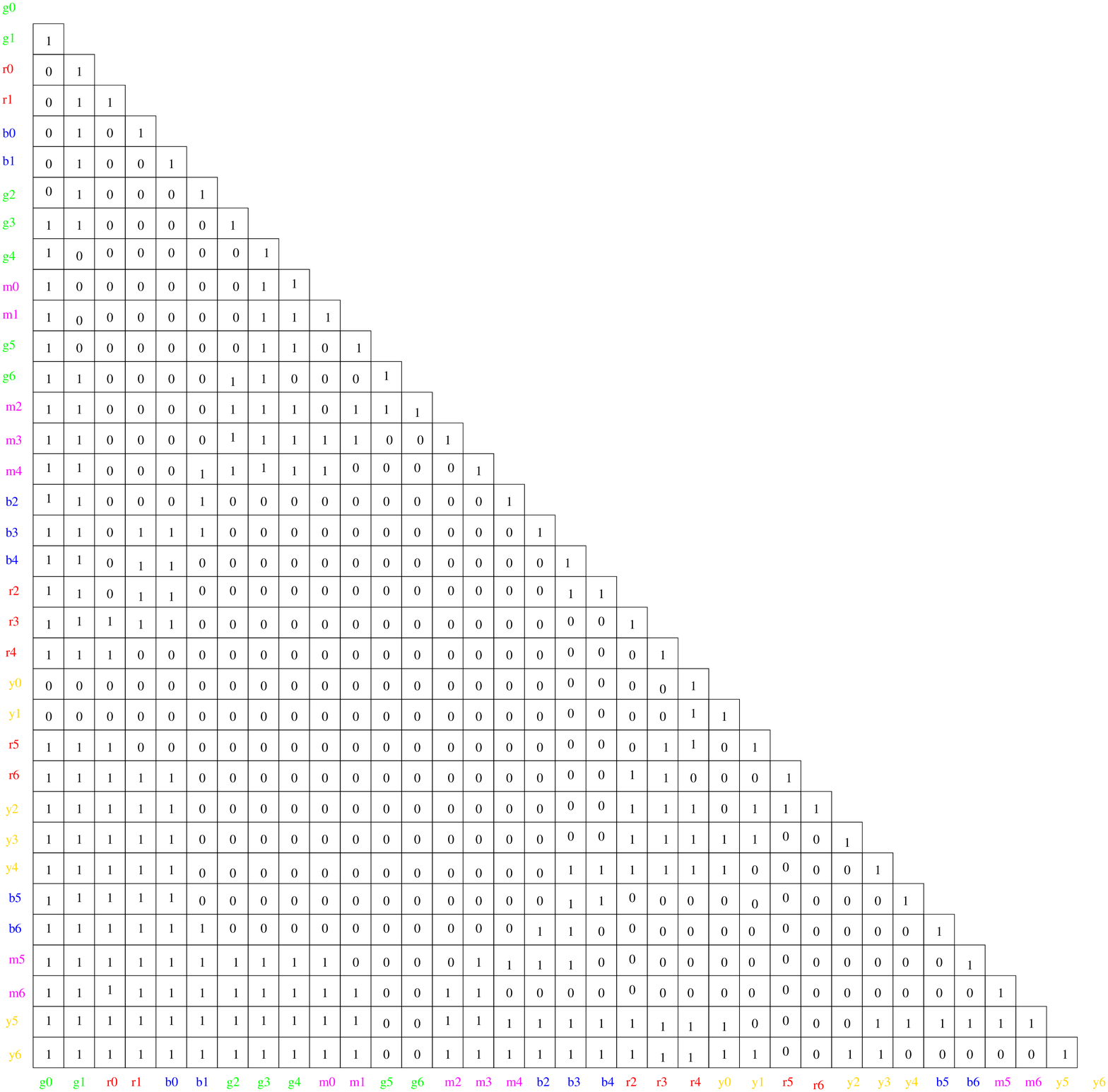}

\caption{\footnotesize The adjacency matrix of $G^*$, a persistent graph that is not a terrain visibility graph. }
\label{fig:adjMatrix}
\end{figure}

Proving that $G^*$ is persistent via a direct proof involves a tedious case analysis, and we instead show it is persistent via a computer program.  The program builds the adjacency matrix as it is shown in Figure \ref{fig:adjMatrix} and then ensures that the graph satisfies both the X-property and the Bar-property.  It can be much more easily verified that the algorithm we used to check the properties is correct than it would be to analyze a direct proof that $G^*$ is persistent.  A copy of the C++ source code we use to perform the check can be found at \url{https://github.com/PySean/GraphChecker}.


The following lemma will be used to prove the main result.

\begin{lemma}
If $X$ satisfies $d_{0,1} d_{2,3} d_{3,4} d_{5,6} > d_{1,2} d_{4,5} d_{0,3} d_{3,6}$, then at least one of the following two statements is true: 1) $d_{1,2} < \min\{d_{0,1},d_{2,3}\}$, or 2) $d_{4,5} < \min\{d_{3,4},d_{5,6}\}$.
\label{lem:helper}
\end{lemma}

\begin{proof}
Suppose without loss of generality that $X$ is such that $d_{1,2} \geq d_{0,1}$ and $d_{4,5} \geq d_{5,6}$.  We will show that $d_{0,1} d_{2,3} d_{3,4} d_{5,6} < d_{1,2} d_{4,5} d_{0,3} d_{3,6}$.  We have:

\begin{align*}
 d_{0,1} d_{2,3} d_{3,4} d_{5,6} &\leq d_{1,2}d_{2,3}d_{3,4}d_{4,5}\\
 &<d_{1,2}d_{0,3}d_{3,6}d_{4,5}\\
 &=d_{1,2} d_{4,5} d_{0,3} d_{3,6}
\end{align*}

Note that the second inequality follows since for all $X$ we have $x_0 < x_2 < x_3$ implying $d_{2,3} < d_{0,3}$, and similarly we have $d_{3,4} < d_{3,6}$.  

The lemma follows by applying a similar analysis for the other 3 cases.  For example, if $d_{1,2} \geq d_{2,3}$ and $d_{4,5} \geq d_{3,4}$ then we'd have:

\begin{align*}
 d_{0,1} d_{2,3} d_{3,4} d_{5,6} &\leq d_{0,1}d_{1,2}d_{4,5}d_{5,6}\\
 &<d_{0,3}d_{1,2}d_{4,5}d_{3,6}\\
 &=d_{1,2} d_{4,5} d_{0,3} d_{3,6}
\end{align*}
\end{proof}

For any color $c$ from our set of colors $\{g,r,b,m,y\}$ and any pair of distinct integers $i,j \in \{0,\ldots, 6\}$ such that $i<j$, we let $d^c_{i,j}$ denote the absolute value of the difference of x-coordinates of $c_i$ and $c_j$.  For example, $d^m_{2,3}$ is the absolute value of the difference of x-coordinates of $m_2$ and $m_3$.  We next show that for any vector $X$ of thirty-five, increasing x-coordinates, at least one color class has to violate the inequality from Lemma \ref{lem:mainIneq}.

\begin{lemma}
Let $X$ be any vector of 35 x-coordinates in increasing order.  There is at least one color $c \in \{g,r,b,m,y\}$ such that the x-coordinates for the seven points of that color do not satisfy $d^c_{0,1} d^c_{2,3} d^c_{3,4} d^c_{5,6} > d^c_{1,2} d^c_{4,5} d^c_{0,3} d^c_{3,6}$.
\label{lem:noX}
\end{lemma}
\begin{proof}
If blue does not satisfy the inequality then we are done, so suppose that blue does satisfy it.  Then according to Lemma \ref{lem:helper}, it must be that either $d^b_{1,2} < d^b_{0,1}$ or $d^b_{4,5} < d^b_{5,6}$.  Without loss of generality, suppose that $d^b_{1,2} < d^b_{0,1}$.  

Now consider the green points.  If green does not satisfy the inequality then we are done, so suppose it does.  Since $g_1 < b_0 < b_1 < g_2 < g_3 < b_2$ and $d^b_{1,2} < d^b_{0,1},$ we must have that $d^g_{2,3} < d^g_{1,2}$.  Then by Lemma \ref{lem:helper} we have that that $d^g_{4,5} < d^g_{5,6}$.

Now consider the magenta points.  If magenta does not satisfy the inequality then we are done, so suppose it does.  Since $g_4 < m_0 < m_1 < g_5 < g_6 < m_2$ and $d^g_{4,5} < d^g_{5,6}$, we have that $d^m_{0,1} < d^m_{1,2}$.  Therefore if magenta satisfies the inequality then we have $d^m_{4,5} < d^m_{3,4}$ and $d^m_{4,5} < d^m_{5,6}$ by Lemma \ref{lem:helper}.  

Now consider the red points.  If red does not satisfy the inequality then we are done, so suppose it does.  Since $r_1 < m_3 < m_4 < r_2 < r_3 < m_5$ and $d^m_{4,5} < d^m_{3,4}$, we have $d^r_{2,3} < d^r_{1,2}$.  Then by Lemma \ref{lem:helper}, we must have that $d^r_{4,5} < d^r_{5,6}$.

Now consider the yellow points.  Since $m_4 < y_3 < y_4 < m_5 < m_6 < y_5$ and $d^m_{4,5} < d^m_{5,6}$ then it must be that $d^y_{3,4} < d^y_{4,5}$.  Since $r_4 < y_0 < y_1 < r_5 < r_6 < y_2$ and $d^r_{4,5} < d^r_{5,6}$, we also have that $d^y_{0,1} < d^y_{1,2}$.  Then by Lemma \ref{lem:helper} we have that yellow must violate the inequality.
\end{proof}

We now show that $G^*$ is not the visibility graph for any terrain, proving Theorem \ref{thm:mainResult}.

\begin{lemma}
For any choice $X$ of thirty-five, increasing x-coordinates, $\mathcal{T}(G^*,X) = \emptyset$.
\end{lemma}

\begin{proof}
By Lemma \ref{lem:noX}, there must be at least one color that does not satisfy the inequality from Lemma \ref{lem:mainIneq}.  Arbitrarily pick one such color with a violated inequality, and let $c$ denote our choice. 

Let $\textbf{A}$ be the constraint matrix generated by our reconstruction approach for $G^*$.  Note that for each of the 6 constraints that we used in the proof of Lemma \ref{lem:mainIneq}, we must have a similar set of constraints for the points of color $c$ here, namely: (1) $p^c_1$ should be above $\overline{p^c_0p^c_2}$, (2) $p^c_3$ should be under $\overline{p^c_0p^c_4}$, (3) $p^c_3$ should be over $\overline{p^c_1p^c_5}$, (4) $p^c_3$ should be under $\overline{p^c_2p^c_6}$, (5) $p^c_5$ should be under $\overline{p^c_3p^c_6}$, and (6) $p^c_5$ should be over $\overline{p^c_4p^c_6}$. The ``under" constraints clearly must be satisfied, but it is not immediately clear that the ``over" constraints must be satisfied: it must be verified that, for example, $p^c_1$ is a designated blocker for $p^c_0$ and $p^c_2$ (for example, $p^c_0$ shouldn't see any points of any color between $p^c_1$ and $p^c_2$).  One can easily verify that this is the case for $G^*$ for each of the ``over" constraints for each of the color classes.

We then prove that $\mathcal{T}(G^*,X) = \emptyset$ using Farkas' Lemma.  That is, we show the existence of a vector $\textbf{z}$ such that $\textbf{z} \leq 0$, $\textbf{A}^\text{T}\textbf{z} \geq 0$, and $\textbf{b}^\text{T}\textbf{z} < 0$ for every $\epsilon > 0$. Note that each entry in $\textbf{z}$ corresponds with one of the constraints of $A$.  We can simply pick our $\textbf{z}$ by allowing each of the entries in $\textbf{z}$ that correspond with one of the six constraints associated with the vertices of color $c$ to take the same value as the corresponding entry in our vector in the proof of Lemma \ref{lem:mainIneq}.  We set every other entry of $\textbf{z}$ to be 0.  The analysis to see that this vector satisfies the conditions of Case 2 of Farkas' Lemma is then identical to that of the proof of Lemma \ref{lem:mainIneq}, completing the proof of this lemma.
\end{proof}


\section{Conclusions and open problems}
\label{sec:conc}
The visibility graphs of terrains have been studied for almost 30 years, and it was known that the visibility graph for any terrain must be persistent.  Previous works tended to believe that persistence formed a characterization of the visibility graphs of terrains, that is that for any persistent graph $G$, there is a terrain $T$ such that $G$ is the visibility graph of $T$.  Our main result in this paper is to show the existence of a persistent graph that is \textit{not} the visibility graph for any terrain.  This proves that pseudo-terrains are not stretchable (as every persistent graph is the visibility graph for some pseudo-terrain).  

There is much left to be determined about the visibility graphs of terrains.  This paper re-opens the question about obtaining a characterization of the visibility graphs of terrains.  We now have that the X-property and Bar-properties are necessary but not sufficient properties for a graph to be the visibility graph of a terrain.  What additional properties must the graph satisfy?  We believe our linear programming approach to reconstructing terrains can shed some light on the reconstruction problem as well.  Previous research attempted to perform an iterative placement of points from left to right.  Our work shows that one needs not be concerned with the y-coordinates of points when reconstructing a terrain, as if one has a set of feasible x-coordinates then the y-coordinates can be computed in polynomial time using linear programming.  Given a visibility graph for a terrain, is there a polynomial-time algorithm that can compute such a set of x-coordinates?

\bibliographystyle{plain}
\bibliography{proposal.bib}

\appendix


\section{Linear algebra proofs for Lemma \ref{lem:mainIneq}}
\label{sec:linearAlgebra}

Recall that in Lemma \ref{lem:mainIneq}, we argue that there is a terrain in $\mathcal{T}(G',X)$ if and only if $X$ is such that $d_{0,1}d_{2,3}d_{3,4}d_{5,6} > d_{1,2}d_{4,5}d_{0,3}d_{3,6}$.  In this section, we show that $\textbf{Ay} \geq \textbf{b}$ when $d_{0,1}d_{2,3}d_{3,4}d_{5,6} > d_{1,2}d_{4,5}d_{0,3}d_{3,6}$, and we show $\textbf{A}^\text{T}\textbf{z} = \textbf{0}$ when $d_{0,1}d_{2,3}d_{3,4}d_{5,6} \leq d_{1,2}d_{4,5}d_{0,3}d_{3,6}$ as claimed in the proof of Lemma \ref{lem:mainIneq}.

\subsection{$\textbf{Ay} \geq \textbf{b}$}
\label{sec:primal}

In this subsection we are assuming $d_{0,1}d_{2,3}d_{3,4}d_{5,6} > d_{1,2}d_{4,5}d_{0,3}d_{3,6}$.  Recall that $\textbf{A}, \textbf{y}$, and $\textbf{b}$ are defined as follows:

\begin{center}
\begin{tabular}{c c}
$\textbf{A} = 
\begin{vmatrix}
-d_{1,2} & d_{0,2} & -d_{0,1}  & 0   & 0 & 0 & 0 \\
d_{3,4} & 0          & 0    & -d_{0,4}      & d_{0,3}     & 0 & 0\\
0    & -d_{3,5}    & 0         & d_{1,5}    & 0    & -d_{1,3} & 0   \\
0     & 0 & d_{3,6} & -d_{2,6} & 0        & 0     & d_{2,3}\\
0     & 0    & 0         & d_{5,6}   & 0   & -d_{3,6} & d_{3,5}   \\
0     & 0          & 0    & 0          & -d_{5,6} & d_{4,6} & -d_{4,5}
\end{vmatrix}  $

&
$\textbf{b} = 
\begin{vmatrix}
\epsilon \\
\epsilon \\
\epsilon \\
\epsilon \\
\epsilon \\
\epsilon 
\end{vmatrix}  $ \\[40pt]


\multicolumn{2}{c}{$\textbf{y} =
    \begin{vmatrix}
    d_{0,1}d_{2,3}d_{3,5}d_{5,6} + d_{0,3}d_{3,5}d_{0,1}d_{2,3} + 
    d_{0,3}d_{3,5}d_{0,1}d_{4,5} +
    d_{4,5}d_{0,3}d_{0,2}d_{3,6} +
    d_{0,3}d_{3,5}d_{4,5}d_{3,6}\\
    \\
    
    d_{1,2}d_{4,5}d_{0,3}d_{3,6} - d_{0,1}d_{2,3}d_{3,4}d_{5,6} \\
    \\
    
    -d_{2,3}d_{5,6}(d_{3,4}d_{0,2} + d_{3,5}(d_{1,2} + d_{3,4}))
    - d_{1,2}d_{0,3}d_{3,5}(d_{2,3} + d_{4,5})\\
    \\
    
    0\\
    \\
    
    -d_{0,1}d_{3,4}d_{3,5}(d_{2,3} + d_{4,5}) - d_{4,5}d_{3,6}
    (d_{3,4}(d_{0,2}+d_{3,5}) + d_{1,2}d_{3,5})\\
    \\
    0\\
    \\
    
    d_{0,1}d_{3,4}d_{5,6}d_{3,5} + d_{3,4}d_{5,6}d_{0,2}d_{3,6} + d_{1,2}d_{5,6}d_{3,5}d_{3,6} + d_{3,4}d_{5,6}d_{3,5}d_{3,6} + d_{1,2}d_{0,3}d_{3,6}d_{3,5}
    \end{vmatrix}   $}
\end{tabular}
\end{center}

Here, $\epsilon$ is the minimum of $d_{3,5}(d_{0,1}d_{2,3}d_{3,4}d_{5,6} - d_{1,2}d_{4,5}d_{0,3}d_{3,6})$ and $d_{3,5}(d_{0,1}d_{3,4}d_{5,6}d_{3,5} + d_{3,4}d_{5,6}d_{0,2}d_{3,6}
    + d_{1,2}d_{5,6}d_{3,5}d_{3,6} + d_{3,4}d_{5,6}d_{3,5}d_{3,6}
    + d_{1,2}d_{0,3}d_{3,6}d_{3,5})$.  Due to our assumption on $X$, $\epsilon$ is always strictly positive and therefore showing $\textbf{Ay} \geq \textbf{b}$ implies that the visibility graph of $T(X,\textbf{y})$ is $G'$ by Lemma \ref{lem:lpFeas}.
    
In this section, we let $R_1, \ldots, R_6$ denote the rows of $A$.  We show the result of $R_i\textbf{y} = d_{3,5}(d_{0,1}d_{2,3}d_{3,4}d_{5,6} - d_{1,2}d_{4,5}d_{0,3}d_{3,6})$ for each $i \neq 5$, and we show that $R_5\textbf{y} = d_{3,5}(d_{0,1}d_{3,4}d_{5,6}d_{3,5} + d_{3,4}d_{5,6}d_{0,2}d_{3,6}
    + d_{1,2}d_{5,6}d_{3,5}d_{3,6} + d_{3,4}d_{5,6}d_{3,5}d_{3,6}
    + d_{1,2}d_{0,3}d_{3,6}d_{3,5})$.  This implies that $R_i\textbf{y} \geq \epsilon$ for all rows.
    
\begin{align*}
 R_1\textbf{y} &= -d_{1,2}[d_{0,1}d_{2,3}d_{3,5}d_{5,6} + d_{0,3}d_{3,5}d_{0,1}d_{2,3} + 
    d_{0,3}d_{3,5}d_{0,1}d_{4,5} +
    d_{4,5}d_{0,3}d_{0,2}d_{3,6} +
    d_{0,3}d_{3,5}d_{4,5}d_{3,6}]\\ &+ d_{0,2}[d_{1,2}d_{4,5}d_{0,3}d_{3,6} - d_{0,1}d_{2,3}d_{3,4}d_{5,6}] \\&- d_{0,1}[-d_{2,3}d_{5,6}(d_{3,4}d_{0,2} + d_{3,5}(d_{1,2} + d_{3,4}))
    - d_{1,2}d_{0,3}d_{3,5}(d_{2,3} + d_{4,5})]\\[20pt]
    &= -d_{1,2}d_{0,1}d_{2,3}d_{3,5}d_{5,6} - d_{1,2}d_{0,3}d_{3,5}d_{0,1}d_{2,3} - d_{1,2}d_{0,3}d_{3,5}d_{0,1}d_{4,5} - d_{1,2}d_{4,5}d_{0,3}d_{0,2}d_{3,6}\\ &- d_{1,2}d_{0,3}d_{3,5}d_{4,5}d_{3,6} - d_{0,2}d_{0,1}d_{2,3}d_{3,4}d_{5,6} + d_{0,2}d_{1,2}d_{4,5}d_{0,3}d_{3,6} + d_{0,1}d_{2,3}d_{5,6}d_{3,4}d_{0,2}\\ &+ 
    d_{0,1}d_{2,3}d_{5,6}d_{3,5}d_{1,2} +
    d_{0,1}d_{2,3}d_{5,6}d_{3,5}d_{3,4}+
    d_{0,1}d_{1,2}d_{0,3}d_{3,5}d_{2,3} +
    d_{0,1}d_{1,2}d_{0,3}d_{3,5}d_{4,5}\\[20pt]
    &= -d_{1,2}d_{0,3}d_{3,5}d_{4,5}d_{3,6} + d_{0,1}d_{2,3}d_{5,6}d_{3,5}d_{3,4}\\[20pt]
    &=d_{3,5}(d_{0,1}d_{2,3}d_{3,4}d_{5,6} - d_{1,2}d_{4,5}d_{0,3}d_{3,6})
\end{align*}

\begin{align*}
    R_2 \textbf{y} &= d_{3,4}[d_{0,1}d_{2,3}d_{3,5}d_{5,6} + d_{0,3}d_{3,5}d_{0,1}d_{2,3} + 
    d_{0,3}d_{3,5}d_{0,1}d_{4,5} +
    d_{4,5}d_{0,3}d_{0,2}d_{3,6} +
    d_{0,3}d_{3,5}d_{4,5}d_{3,6}] - d_{0,4}\cdot 0\\ &+ d_{0,3}[-d_{0,1}d_{3,4}d_{3,5}(d_{2,3} + d_{4,5}) - d_{4,5}d_{3,6}
    (d_{3,4}(d_{0,2}+d_{3,5}) + d_{1,2}d_{3,5})]\\[20pt]
    &=d_{3,4}d_{0,1}d_{2,3}d_{3,5}d_{5,6} + d_{3,4}d_{0,3}d_{3,5}d_{0,1}d_{2,3} +
    d_{3,4}d_{0,3}d_{3,5}d_{0,1}d_{4,5} +
    d_{3,4}d_{4,5}d_{0,3}d_{0,2}d_{3,6}\\ &+
    d_{3,4}d_{0,3}d_{3,5}d_{4,5}d_{3,6} 
    - d_{0,3}d_{0,1}d_{3,4}d_{3,5}d_{2,3} - d_{0,3}d_{0,1}d_{3,4}d_{3,5}d_{4,5} - d_{0,3}d_{4,5}d_{3,6}d_{3,4}d_{0,2}\\ &- d_{0,3}d_{4,5}d_{3,6}d_{3,4}d_{3,5} - d_{0,3}d_{4,5}d_{3,6}d_{1,2}d_{3,5}\\[20pt]
    &=d_{3,4}d_{0,1}d_{2,3}d_{3,5}d_{5,6} - d_{0,3}d_{4,5}d_{3,6}d_{1,2}d_{3,5}\\[20pt]
    &=d_{3,5}(d_{0,1}d_{2,3}d_{3,4}d_{5,6} - d_{1,2}d_{4,5}d_{0,3}d_{3,6})
\end{align*}

\begin{align*}
    R_3\textbf{y} &= -d_{3,5}[d_{1,2}d_{4,5}d_{0,3}d_{3,6} - d_{0,1}d_{2,3}d_{3,4}d_{5,6}] + d_{2,5}\cdot 0 - d_{1,3}\cdot 0\\[20pt]
    &= d_{3,5}(d_{0,1}d_{2,3}d_{3,4}d_{5,6} - d_{1,2}d_{4,5}d_{0,3}d_{3,6})
\end{align*}

\begin{align*}
    R_4  \textbf{y} &= d_{3,6}[-d_{2,3}d_{5,6}(d_{3,4}d_{0,2} + d_{3,5}(d_{1,2} + d_{3,4}))
    - d_{1,2}d_{0,3}d_{3,5}(d_{2,3} + d_{4,5})] - d_{2,6}\cdot 0\\ &+ d_{2,3}[d_{0,1}d_{3,4}d_{5,6}d_{3,5} + d_{3,4}d_{5,6}d_{0,2}d_{3,6} + d_{1,2}d_{5,6}d_{3,5}d_{3,6} + d_{3,4}d_{5,6}d_{3,5}d_{3,6} + d_{1,2}d_{0,3}d_{3,6}d_{3,5}]\\[20pt]
    &= d_{2,3}d_{0,1}d_{3,4}d_{5,6}d_{3,5} + d_{2,3}d_{3,4}d_{5,6}d_{0,2}d_{3,6} + d_{2,3}d_{1,2}d_{5,6}d_{3,5}d_{3,6} +
    d_{2,3}d_{3,4}d_{5,6}d_{3,5}d_{3,6}\\ & +
    d_{2,3}d_{1,2}d_{0,3}d_{3,6}d_{3,5}
    -d_{3,6}d_{2,3}d_{5,6}d_{3,4}d_{0,2} - d_{3,6}d_{2,3}d_{5,6}d_{3,5}d_{1,2} - d_{3,6}d_{2,3}d_{5,6}d_{3,5}d_{3,4}\\& - d_{3,6}d_{1,2}d_{0,3}d_{3,5}d_{2,3} - d_{3,6}d_{1,2}d_{0,3}d_{3,5}d_{4,5}\\[20pt]
    &=d_{2,3}d_{0,1}d_{3,4}d_{5,6}d_{3,5} - d_{3,6}d_{1,2}d_{0,3}d_{3,5}d_{4,5}\\[20pt]
    &= d_{3,5}(d_{0,1}d_{2,3}d_{3,4}d_{5,6} - d_{1,2}d_{4,5}d_{0,3}d_{3,6})
\end{align*}

\begin{align*}
    R_5  \textbf{y} &=
    d_{5,6}\cdot 0 - d_{3,6}\cdot 0 \\&+ d_{3,5}[ d_{0,1}d_{3,4}d_{5,6}d_{3,5} + d_{3,4}d_{5,6}d_{0,2}d_{3,6} + d_{1,2}d_{5,6}d_{3,5}d_{3,6} + d_{3,4}d_{5,6}d_{3,5}d_{3,6} + d_{1,2}d_{0,3}d_{3,6}d_{3,5}]\\[20pt]
    &= d_{3,5}( d_{0,1}d_{3,4}d_{5,6}d_{3,5} + d_{3,4}d_{5,6}d_{0,2}d_{3,6} + d_{1,2}d_{5,6}d_{3,5}d_{3,6} + d_{3,4}d_{5,6}d_{3,5}d_{3,6} + d_{1,2}d_{0,3}d_{3,6}d_{3,5})
\end{align*}

\begin{align*}
    R_6  \textbf{y} &= -(d_{5,6}[-d_{0,1}d_{3,4}d_{3,5}(d_{2,3} + d_{4,5}) - d_{4,5}d_{3,6}(d_{3,4}(d_{0,2}+d_{3,5}) + d_{1,2}d_{3,5})] + d_{4,6}\cdot 0\\ 
    &- d_{4,5}[d_{0,1}d_{3,4}d_{5,6}d_{3,5} + d_{3,4}d_{5,6}d_{0,2}d_{3,6} + d_{1,2}d_{5,6}d_{3,5}d_{3,6} + d_{3,4}d_{5,6}d_{3,5}d_{3,6} + d_{1,2}d_{0,3}d_{3,6}d_{3,5}]\\[20pt]
    &= d_{5,6}d_{0,1}d_{3,4}d_{3,5}d_{2,3} + d_{5,6}d_{0,1}d_{3,4}d_{3,5}d_{4,5} + d_{5,6}d_{4,5}d_{3,6}d_{3,4}d_{0,2} + d_{5,6}d_{4,5}d_{3,6}d_{3,4}d_{3,5}\\ & + d_{5,6}d_{4,5}d_{3,6}d_{1,2}d_{3,5} - d_{4,5}d_{0,1}d_{3,4}d_{5,6}d_{3,5} - 
    d_{4,5}d_{3,4}d_{5,6}d_{0,2}d_{3,6} - 
    d_{4,5}d_{1,2}d_{5,6}d_{3,5}d_{3,6}\\ &-
    d_{4,5}d_{3,4}d_{5,6}d_{3,5}d_{3,6} - 
    d_{4,5}d_{1,2}d_{0,3}d_{3,6}d_{3,5}\\[20pt]
    &= d_{5,6}d_{0,1}d_{3,4}d_{3,5}d_{2,3} - 
    d_{4,5}d_{1,2}d_{0,3}d_{3,6}d_{3,5}\\[20pt]
    &= d_{3,5}(d_{0,1}d_{2,3}d_{3,4}d_{5,6} - d_{1,2}d_{4,5}d_{0,3}d_{3,6})
\end{align*}

\subsection{$\textbf{A}^\text{T}\textbf{z} = \textbf{0}$}
\label{sec:dual}

In this subsection we are assuming $d_{0,1}d_{2,3}d_{3,4}d_{5,6} \leq d_{1,2}d_{4,5}d_{0,3}d_{3,6}$.  Recall that $\textbf{A}^\text{T}$ and $\textbf{z}$ are defined as follows:

\begin{center}
\begin{tabular}{c c}
$\textbf{A}^\text{T} = 
\begin{vmatrix}
-d_{1,2}    & d_{3,4}          & 0         & 0          & 0        & 0     \\
d_{0,2} & 0          & -d_{3,5}    & 0          & 0        & 0     \\
-d_{0,1}    & 0          & 0         & d_{3,6}    & 0        & 0     \\
0     & -d_{0,4} & d_{1,5} & -d_{2,6} & d_{5,6}        & 0     \\
0     & d_{0,3}    & 0         & 0          & 0        & -d_{5,6}    \\
0     & 0          & -d_{1,3}    & 0          & -d_{3,6} & d_{4,6} \\
0     & 0          & 0         & d_{2,3}          & d_{3,5}    & -d_{4,5}
\end{vmatrix}    $ 

&
$\textbf{z} =
    \begin{vmatrix}
    \dfrac{-d_{3,4}d_{5,6}}{d_{1,2}d_{0,3}} \\
    \\
    
    \dfrac{-d_{5,6}}{d_{0,3}} \\
    \\
    
    \dfrac{-d_{3,4}d_{5,6}d_{0,2}}{d_{2,1}d_{0,3}d_{3,5}} \\
    \\
    
    \dfrac{-d_{0,1}d_{3,4}d_{5,6}}{d_{1,2}d_{0,3}d_{3,6}}\\
    \\
    \dfrac{d_{0,1}d_{2,3}d_{3,4}d_{5,6} - d_{1,2}d_{4,5}d_{0,3}d_{3,6}}{d_{1,2}d_{0,3}d_{3,6}d_{3,5}}\\
    \\
    -1
    \end{vmatrix}   $ 
\end{tabular}
\end{center}

In this section, we let $R_1, \ldots, R_7$ denote the rows of $\textbf{A}^\text{T}$.  We show the result of $R_i \textbf{z} = 0$ for each row.

\begin{align*}
R_1  Z &= \dfrac{d_{1,2}d_{3,4}d_{5,6}}{d_{1,2}d_{0,3}} + \dfrac{-d_{3,4}d_{5,6}}{d_{0,3}} \\[20pt]
            &= \dfrac{d_{3,4}d_{5,6}}{d_{0,3}} + \dfrac{-d_{3,4}d_{5,6}}{d_{0,3}} \\[20pt]
            &= 0
\end{align*}

\begin{align*}
R_2  Z &= \dfrac{-d_{3,4}d_{5,6}d_{0,2}}{d_{1,2}d_{0,3}} + \dfrac{d_{3,4}d_{5,6}d_{0,2}d_{3,5}}{d_{1,2}d_{0,3}d_{3,5}} \\[20pt]
            &= \dfrac{-d_{3,4}d_{5,6}d_{0,2}}{d_{1,2}d_{0,3}} + \dfrac{d_{3,4}d_{5,6}d_{0,2}}{d_{1,2}d_{0,3}} \\[20pt]
            &= 0 
            \end{align*}

\begin{align*}
R_3  Z &= \dfrac{d_{0,1}d_{3,4}d_{5,6}}{d_{1,2}d_{0,3}} - \dfrac{d_{0,1}d_{3,4}d_{5,6}d_{3,6}}{d_{1,2}d_{0,3}d_{3,6}} \\[20pt]
            &= \dfrac{d_{0,1}d_{3,4}d_{5,6}}{d_{1,2}d_{0,3}} - \dfrac{d_{0,1}d_{3,4}d_{5,6}}{d_{1,2}d_{0,3}} \\[20pt]
            &= 0 
            \end{align*}

\begin{align*}
R_4  Z &= \dfrac{d_{5,6}d_{0,4}}{d_{0,3}} - \dfrac{d_{3,4}d_{5,6}d_{0,2}d_{1,5}}{d_{1,2}d_{0,3}d_{3,5}}
            +  \dfrac{d_{0,1}d_{3,4}d_{5,6}d_{2,6}}{d_{1,2}d_{0,3}d_{3,6}}+ 
            \dfrac{d_{5,6}(d_{0,1}d_{2,3}d_{3,4}d_{5,6} - d_{1,2}d_{4,5}d_{0,3}d_{3,6})}{d_{1,2}d_{0,3}d_{3,6}d_{3,5}}\\[20pt]
            &= d_{5,6}\Bigg(\dfrac{d_{0,4}}{d_{0,3}} - \dfrac{d_{3,4}d_{0,2}d_{1,5}}{d_{1,2}d_{0,3}d_{3,5}}
            + \dfrac{d_{0,1}d_{3,4}d_{2,6}}{d_{1,2}d_{0,3}d_{3,6}} +
            \dfrac{d_{0,1}d_{2,3}d_{3,4}d_{5,6} - d_{1,2}d_{4,5}d_{0,3}d_{3,6}}{d_{1,2}d_{0,3}d_{3,6}d_{3,5}}\Bigg)\\[20pt]
            &=
            d_{5,6}\Bigg(\dfrac{d_{0,4}d_{1,2}d_{3,5} - d_{0,2}d_{1,5}d_{3,4}}{d_{0,3}d_{1,2}d_{3,5}} + \dfrac{d_{0,1}d_{3,4}d_{2,6}}{d_{1,2}d_{0,3}d_{3,6}} 
            +
            \dfrac{d_{0,1}d_{2,3}d_{3,4}d_{5,6} - d_{1,2}d_{4,5}d_{0,3}d_{3,6}}{d_{1,2}d_{0,3}d_{3,6}d_{3,5}}\Bigg)\\[20pt]
            &= d_{5,6}\Bigg(\dfrac{d_{1,2}(d_{0,1}+d_{1,2}+d_{2,3}+d_{3,4})(d_{3,4}+d_{4,5}) - d_{3,4}(d_{0,1}+d_{1,2})(d_{1,2}+d_{2,3}+d_{3,4}+d_{4,5})}{d_{0,3}d_{1,2}d_{3,5}} \\ &+ \dfrac{d_{0,1}d_{3,4}d_{2,6}}{d_{1,2}d_{0,3}d_{3,6}} 
           +
            \dfrac{d_{0,1}d_{2,3}d_{3,4}d_{5,6} - d_{1,2}d_{4,5}d_{0,3}d_{3,6}}{d_{1,2}d_{0,3}d_{3,6}d_{3,5}}\Bigg)
            \\[20pt]
            &= d_{5,6}\Bigg(\dfrac{d_{1,2}(d_{0,1}d_{3,4} + d_{0,1}d_{4,5}+ d_{1,2}d_{3,4} + d_{1,2}d_{4,5} + d_{2,3}d_{3,4} + d_{2,3}d_{4,5} + d_{3,4}d_{3,4} + d_{3,4}d_{4,5})}{d_{0,3}d_{1,2}d_{3,5}} 
            \\
            &- \dfrac{d_{3,4}(d_{0,1}d_{1,2} + d_{0,1}d_{2,3} + d_{0,1}d_{3,4} + d_{0,1}d_{4,5} + d_{1,2}d_{1,2} + d_{1,2}d_{2,3} + d_{1,2}d_{3,4} + d_{1,2}d_{4,5})}{d_{0,3}d_{1,2}d_{3,5}} \\ &+ \dfrac{d_{0,1}d_{3,4}d_{2,6}}{d_{1,2}d_{0,3}d_{3,6}} 
           +
            \dfrac{d_{0,1}d_{2,3}d_{3,4}d_{5,6} - d_{1,2}d_{4,5}d_{0,3}d_{3,6}}{d_{1,2}d_{0,3}d_{3,6}d_{3,5}}\Bigg)\end{align*}
{\setlength{\mathindent}{0cm}\begin{align*}
            &= d_{5,6}\Bigg(\dfrac{d_{1,2}d_{0,1}d_{4,5} + d_{1,2}d_{1,2}d_{4,5} + d_{1,2}d_{2,3}d_{4,5} - d_{3,4}d_{0,1}d_{2,3} - d_{3,4}d_{0,1}d_{3,4} - d_{3,4}d_{0,1}d_{4,5}}{d_{0,3}d_{1,2}d_{3,5}} 
            \\
            &+ \dfrac{d_{0,1}d_{3,4}d_{2,6}}{d_{1,2}d_{0,3}d_{3,6}} 
            +
            \dfrac{d_{0,1}d_{2,3}d_{3,4}d_{5,6} - d_{1,2}d_{4,5}d_{0,3}d_{3,6}}{d_{1,2}d_{0,3}d_{3,6}d_{3,5}}\Bigg)
            \\[20pt]
            &= d_{5,6}\Bigg(\dfrac{d_{1,2}d_{4,5}(d_{0,1} + d_{1,2} + d_{2,3}) - d_{0,1}d_{3,4}(d_{2,3} + d_{3,4} + d_{4,5})}{d_{0,3}d_{1,2}d_{3,5}} 
            \\
            &+ \dfrac{d_{0,1}d_{3,4}d_{2,6}}{d_{1,2}d_{0,3}d_{3,6}} 
            +
            \dfrac{d_{0,1}d_{2,3}d_{3,4}d_{5,6} - d_{1,2}d_{4,5}d_{0,3}d_{3,6}}{d_{1,2}d_{0,3}d_{3,6}d_{3,5}}\Bigg)
            \\[20pt]
            &= d_{5,6}\Bigg(\dfrac{d_{3,6}d_{1,2}d_{4,5}d_{0,3} - d_{0,1}d_{3,4}(d_{2,3} + d_{3,4} + d_{4,5})(d_{3,4} + d_{4,5} + d_{5,6})}{d_{0,3}d_{1,2}d_{3,5}d_{3,6}} 
            \\
            &+ \dfrac{d_{0,1}d_{3,4}(d_{2,3} + d_{3,4} + d_{4,5} + d_{5,6})(d_{3,4} + d_{4,5})}{d_{0,3}d_{1,2}d_{3,5}d_{3,6}} 
            +
            \dfrac{d_{0,1}d_{2,3}d_{3,4}d_{5,6} - d_{1,2}d_{4,5}d_{0,3}d_{3,6}}{d_{1,2}d_{0,3}d_{3,6}d_{3,5}}\Bigg)
            \\[20pt]
            &= d_{5,6}\Bigg(\dfrac{d_{0,1}d_{3,4}(d_{2,3}d_{3,4} + d_{2,3}d_{4,5} + d_{3,4}d_{3,4} + d_{3,4}d_{4,5} + d_{4,5}d_{3,4} + d_{4,5}d_{4,5} + d_{5,6}d_{3,4} + d_{5,6}d_{4,5})}{d_{0,3}d_{1,2}d_{3,5}d_{3,6}}
            \\
            &+ \frac{d_{3,6}d_{1,2}d_{4,5}d_{0,3} - d_{0,1}d_{3,4}(d_{2,3}d_{3,4} + d_{2,3}d_{4,5} + d_{2,3}d_{5,6} + d_{3,4}d_{3,4} + d_{3,4}d_{4,5} + d_{3,4}d_{5,6} {+} d_{4,5}d_{3,4} {+} d_{4,5}d_{4,5} {+} d_{4,5}d_{5,6})}{d_{0,3}d_{1,2}d_{3,5}d_{3,6}}  
            \\
            &+
            \dfrac{d_{0,1}d_{2,3}d_{3,4}d_{5,6} - d_{1,2}d_{4,5}d_{0,3}d_{3,6}}{d_{1,2}d_{0,3}d_{3,6}d_{3,5}}\Bigg)
            \\[20pt]
            &= d_{5,6}\Bigg(\dfrac{d_{3,6}d_{1,2}d_{4,5}d_{0,3} - d_{0,1}d_{3,4}d_{2,3}d_{5,6}}{d_{0,3}d_{1,2}d_{3,5}d_{3,6}} 
           +
            \dfrac{d_{0,1}d_{2,3}d_{3,4}d_{5,6} - d_{1,2}d_{4,5}d_{0,3}d_{3,6}}{d_{1,2}d_{0,3}d_{3,6}d_{3,5}}\Bigg)
            \\[20pt]
            &= 0 \\
\end{align*}}

\begin{align*}
R_5  Z &= -\dfrac{d_{5,6}d_{0,3}}{d_{0,3}} + d_{5,6}\\[20pt]
&= 0
\end{align*}

\begin{align*}
R_6  Z &= \dfrac{d_{1,3}d_{3,4}d_{5,6}d_{0,2}}{d_{1,2}d_{0,3}d_{3,5}} - \dfrac{d_{3,6}(d_{0,1}d_{2,3}d_{3,4}d_{5,6} - d_{1,2}d_{4,5}d_{0,3}d_{3,6})}{d_{1,2}d_{0,3}d_{3,6}d_{3,5}} - d_{4,6}
\\[20pt]
&= \dfrac{d_{3,4}d_{5,6}(d_{1,2} + d_{2,3})(d_{0,1} + d_{1,2})}{d_{1,2}d_{0,3}d_{3,5}} -\dfrac{d_{0,1}d_{2,3}d_{3,4}d_{5,6} - d_{1,2}d_{4,5}d_{0,3}d_{3,6}}{d_{1,2}d_{0,3}d_{3,5}}
- \dfrac{d_{4,6}d_{1,2}d_{0,3}d_{3,5}}{d_{1,2}d_{0,3}d_{3,5}}
\\[20pt]
&= \dfrac{1}{d_{1,2}d_{0,3}d_{3,5}} \Bigg(d_{3,4}d_{5,6}(d_{0,1}d_{1,2} + d_{0,1}d_{2,3} + d_{1,2}d_{1,2} + d_{1,2}d_{2,3}) -  
d_{0,1}d_{2,3}d_{3,4}d_{5,6} \\
&+ d_{1,2}d_{4,5}(d_{0,1}d_{3,4} + d_{0,1}d_{4,5} + d_{0,1}d_{5,6} + d_{1,2}d_{3,4} + d_{1,2}d_{4,5} + d_{1,2}d_{5,6} + d_{2,3}d_{3,4} + d_{2,3}d_{4,5} + d_{2,3}d_{5,6})\\
&- d_{1,2}(d_{4,5}d_{0,1}d_{3,4} + d_{4,5}d_{0,1}d_{4,5} + d_{4,5}d_{1,2}d_{3,4} + d_{4,5}d_{1,2}d_{4,5} + d_{4,5}d_{2,3}d_{3,4} + d_{4,5}d_{2,3}d_{4,5} + d_{5,6}d_{0,1}d_{3,4} \\
&+ d_{5,6}d_{0,1}d_{4,5} + d_{5,6}d_{1,2}d_{3,4} + d_{5,6}d_{1,2}d_{4,5} + d_{5,6}d_{2,3}d_{3,4} + d_{5,6}d_{2,3}d_{4,5}) \Bigg)
\\[20pt]
&= \dfrac{1}{d_{1,2}d_{0,3}d_{3,5}} \Bigg(d_{3,4}d_{5,6}(d_{0,1}d_{1,2} + d_{0,1}d_{2,3} + d_{1,2}d_{1,2} + d_{1,2}d_{2,3}) - d_{0,1}d_{2,3}d_{3,4}d_{5,6}\\
&-     d_{1,2}(d_{5,6}d_{0,1}d_{3,4} + d_{5,6}d_{1,2}d_{3,4} + d_{5,6}d_{2,3}d_{3,4}) \Bigg)
\\[20pt]
&= \dfrac{1}{d_{1,2}d_{0,3}d_{3,5}} \Bigg(d_{0,1}d_{2,3}d_{3,4}d_{5,6} - d_{0,1}d_{2,3}d_{3,4}d_{5,6} \Bigg)
\\[20pt]
&= 0
\end{align*}

\begin{align*}
R_7  Z &= -\dfrac{d_{0,1}d_{3,4}d_{5,6}d_{2,3} + d_{0,1}d_{2,3}d_{3,4}d_{5,6} - d_{1,2}d_{4,5}d_{0,3}d_{3,6}}{d_{1,2}d_{0,3}d_{3,6}} + d_{4,5}
\\[20pt]
&= -\dfrac{d_{1,2}d_{4,5}d_{0,3}d_{3,6}}{d_{1,2}d_{0,3}d_{3,6}} + d_{4,5}
\\[20pt]
&= -d_{4,5} + d_{4,5} 
\\[20pt]
&= 0
\end{align*}
\end{document}